\pdfoutput=1
\documentclass[a4paper,UKenglish,cleveref, autoref, thm-restate]{lipics-v2021}
\usepackage{stmaryrd}
\usepackage{xifthen}
\usepackage[skip=0pt,belowskip=-10pt]{caption}
\usepackage{enumitem}
\newcommand{\pset}{\mathcal{P}}
\newcommand{\Exp}{\mathsf{Exp}}
\newcommand{\seq}{\mathbin{;}}
\newcommand{\bN}{\mathbb{N}}
\newcommand{\bQ}{\mathbb{Q}^{+}}
\newcommand{\eR}{\overline{\mathbb{R}}}

\newcommand{\e}{\varepsilon}
\newcommand{\acirel}{{\mathrel{\dot\equiv}}}
\newcommand{\aciq}{{\Exp}/{\acirel}}
\newcommand{\gen}[2]{\langle #1 \rangle_{{#2}}}
\newcommand{\set}[2]{\left\{ #1 \ifempty{#2}{}{\mid #2} \right\}}
\newcommand{\sem}[1]{\llbracket #1 \rrbracket }
\providecommand*{\ifempty}[3]{\ifthenelse{\isempty{#1}}{#2}{#3}}

\newcommand{\TT}[2][]{\mathbb{T}\ifempty{#2}{\ifempty{#1}{}{(#1)}}{\ifempty{#1}{(#2)}{(#1,#2)}}}  
\newcommand{\Sub}[1][]{\ifempty{#1}{\mathcal{S}}{\mathcal{S}(#1)}} 
\newcommand{\E}[1][]{\ifempty{#1}{\mathcal{E}}{\mathcal{E}(#1)}} 
\newcommand{\U}[1][]{{\mathcal{U}^{#1}}} 

\newcommand{\Refl}{\textsf{Refl}} 
\newcommand{\Symm}{\textsf{Symm}} 
\newcommand{\Triang}{\textsf{Triang}} 
\newcommand{\Max}{\textsf{Max}} 
\newcommand{\Cont}{\textsf{Cont}} 
\newcommand{\Nexp}{\textsf{NExp}} 
\newcommand{\Subst}{\textsf{Subst}} 
\newcommand{\Cut}{\textsf{Cut}} 
\newcommand{\Assum}{\textsf{Assum}} 

\newcommand{\Top}{\textsf{Top}}

\newcommand{\dPref}[1][\lambda]{\textsf{$#1$-Pref}}

\newcommand{\A}{\mathcal{A}} 
\newcommand{\B}{\mathcal{B}} 

\pdfoutput=1 
\hideLIPIcs  

\usepackage{tikz}
\usetikzlibrary{patterns, intersections, positioning, shapes.geometric, calc,
	backgrounds, fit,
	automata, arrows, decorations.pathreplacing,decorations.pathmorphing}
\usetikzlibrary{cd}
\usepackage{multicol, mathpartir}
\title{A Complete Quantitative Axiomatisation of Behavioural Distance of Regular Expressions} 

\author{Wojciech Różowski}{Department of Computer Science, University College London, United Kingdom \and \url{https://wkrozowski.github.io}}{w.rozowski@cs.ucl.ac.uk}{https://orcid.org/0000-0002-8241-7277}{}

\authorrunning{W. Różowski} 

\Copyright{Wojciech Różowski} 



\category{} 
\relatedversion{}

\funding{This work was partially supported by ERC grant Autoprobe (grant agreement 101002697).}

\acknowledgements{The author wishes to thank Giorgio Bacci, Leo Lobski, Alexandra Silva and Mateo Torres-Ruiz for discussions and feedback, as well as anonymous reviewers for their comments.}

\nolinenumbers 

\EventEditors{Karl Bringmann, Martin Grohe, Gabriele Puppis, and Ola Svensson}
\EventNoEds{4}
\EventLongTitle{51st International Colloquium on Automata, Languages, and Programming (ICALP 2024)}
\EventShortTitle{ICALP 2024}
\EventAcronym{ICALP}
\EventYear{2024}
\EventDate{July 8--12, 2024}
\EventLocation{Tallinn, Estonia}
\EventLogo{}
\SeriesVolume{297}
\ArticleNo{137}

\begin{document}

\maketitle

\begin{abstract}
Deterministic automata have been traditionally studied through the point of view of language equivalence, but another perspective is given by the canonical notion of \emph{shortest-distinguishing-word} distance quantifying the of states. Intuitively, the longer the word needed to observe a difference between two states, then the closer their behaviour is. In this paper, we give a sound and complete axiomatisation of \emph{shortest-distinguishing-word} distance between regular languages. Our axiomatisation relies on a recently developed quantitative analogue of equational logic, allowing to manipulate rational-indexed judgements of the form $e \equiv_\e f$ meaning \emph{term $e$ is approximately equivalent to term $f$ within the error margin of $\e$}. The technical core of the paper is dedicated to the completeness argument that draws techniques from order theory and Banach spaces to simplify the calculation of the behavioural distance to the point it can be then mimicked by axiomatic reasoning.	
\end{abstract}

\section{Introduction}\label{sec: introduction}
Transition systems have been widely employed to model computational phenomena. In theoretical computer science, it is customary to model computations as transition systems and subsequently reason about their equivalence or similarity.
 Classical examples include checking language equivalence of deterministic finite automata using Hopcroft and Karp's algorithm~\cite{Hopcroft:1971:Linear} or constructing bisimulations between labelled transition systems~\cite{Park:1981:Concurrency}. Throughout the years, especially in the concurrency theory community, researchers have studied a plethora of different notions of behavioural equivalences and preorders one could impose on a transition system~\cite{Glabbeek:1990:Linear}. However, in many practical applications, especially when dealing with probabilistic or quantitative transition systems, asking about such classical notions of equivalence (or similarity) could be too strict, and it might be more reasonable to ask quantitative questions about how far apart the behaviour of the two states is.

A growing line of work on \emph{behavioural distances}~\cite{Breugel:2017:Probabilistic,Baldan:2018:Coalgebraic,Breugel:2001:Towards,Breugel:2006:Approximating,Desharnais:2014:Metrics} answers this problem by equipping state-spaces of transition systems with (pseudo)metric structures quantifying the dissimilarity of states. In such a setting, states at distance zero are not necessarily the same, but rather equivalent with respect to some classical notion of behavioural equivalence. In a nutshell, equipping transition systems with such a notion of distance crucially relies on the possibility of \emph{lifting} the distance between the states to the distance on the observable behaviour of the transition system. Behavioural distances were originally studied in the context of probabilistic transition systems~\cite{Giacalone:1990:Algebraic,Breugel:2001:Towards}, where observable behaviour is in the form of probability distribution among possible transitions. In such a case, the necessary lifting of distances between states to distances between probability distributions of possible outcomes relies on the famous Kantorovich/Wasserstein liftings, studied traditionally in transportation theory~\cite{Villani:2009:Optimal}. In general, transition systems can be viewed more abstractly through a well-established category-theoretic framework of coalgebras for an endofunctor~\cite{Rutten:2000:Universal}. Recent work~\cite{Baldan:2018:Coalgebraic} generalised the Kantorovich/Wassertstein lifting to lifting endofunctors (modelling one-step behaviour of transition systems) from the category of sets and functions to the category of (pseudo)metric spaces and nonexpansive functions, thus allowing equipping a multitude of different kinds of transition systems with a sensible notion of behavioural distance.

Traditionally, besides looking at behavioural equivalence/similarity purely from the algorithmic point of view, one can look at those problems axiomatically, by describing behaviours of transition systems as expressions and by providing formal systems based on (in)equational logic for reasoning about equivalence/similarity of the transition systems described by the expressions. Classic examples include reasoning about language equivalence of Kleene's regular expressions representing deterministic finite automata using inference systems of Salomaa~\cite{Salomaa:1966:Two} or Kozen~\cite{Kozen:1994:Completeness}, or reasoning about bisimilarity of finite-state labelled transition systems through Milner's calculus of finite-state behaviours~\cite{Milner:1984:Complete}. 

In this paper, we are interested in a similar axiomatic point of view, but in the case of behavioural distances. Unfortunately, the classical (in)equational logic cannot be applied here, as it has no way of talking about approximate equivalence. Instead, we rely on the quantitative analogue of equational logic~\cite{Mardare:2016:Quantitative}, which deals with the statements of the form $e \equiv_{\e} f$, intuitively meaning \emph{term $e$ is within the distance of at most	 $\e \in \bQ$ from the term $f$}. While the existing work~\cite{Bacci:2018:Algebraic,Bacci:2018:Bisimilarity,Bacci:2018:TV} looked at quantitative axiomatisations of behavioural distance for probabilistic transition systems calculated through the Kantorovich/Wasserstein lifting, which can be thought of as a special case of the abstract coalgebraic framework relying on lifting endofunctors to the category of pseudometric spaces, the notions of behavioural distance for other kinds of transition systems have not been axiomatised before. It turns out that the approach to completeness used in~\cite{Bacci:2018:Bisimilarity} relies on properties which are not unique to distances obtained through the Kantorovich/Wasserstein lifting and can be employed to give complete axiomatisations of behavioural distances for other kinds of transition systems obtained through the coalgebraic framework~\cite{Baldan:2018:Coalgebraic}. In this paper, as a starting point, we look at one of the simplest instantiations of that abstract framework in the case of deterministic automata, yielding \emph{shortest-distinguishing-word} distance. To illustrate that notion of distance, consider the following three deterministic finite automata:

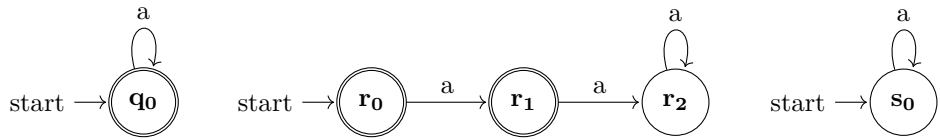
\begin{tikzpicture}[shorten >=1pt,node distance=2cm,on grid,auto,every node/.style={scale=1}]

  \node[initial,state,accepting]  (q_0)                      {$\mathbf{q_0}$};
  \path[->] (q_0) edge [loop above] node {a} ();
  \node [initial, state, accepting] (r_0) [right= 3 cm of q_0] {$\mathbf{r_0}$};
  \node [state,accepting] (r_1) [right= of r_0] {$\mathbf{r_1}$};
  \node [state] (r_2) [right= of r_1] {$\mathbf{r_2}$}; 
  \path[->] (r_0) edge node {a} (r_1);
  \path[->] (r_1) edge node {a} (r_2);
  \path[->] (r_2) edge [loop above] node {a} ();
  \node[initial, state] (s_0) [right= 3 cm of r_2] {$\mathbf{s_0}$};
  \path[->] (s_0) edge [loop above] node {a} ();
\end{tikzpicture}

Neither of the above automata are language equivalent. Their languages are respectively (from the left) $\{\epsilon, a, aa, aaa, \dots\}$, $\{\epsilon, a\}$ and $\emptyset$. However, one could argue that the behaviour of the middle automaton is closer to the one on the left rather than the one on the right. In particular, languages of the left and middle automaton agree on all of the words of length less than two, while the left and right one disagree on all words. 
One can make this idea precise, by providing $1$-bounded metric $d_{\mathcal{L}} : \pset (A^*) \times \pset (A^*) \to [0,1]$ on the set of all formal languages over some fixed alphabet $A$ given by the following formula, where $\lambda \in ]0,1[$ and $L, M \subseteq A^*$:
\begin{equation}\label{eq: shortest_distinguishing_word}
	d_{\mathcal{L}}(L, M) = \begin{cases}
\lambda^{|w|} & w \text{ is the shortest word that belongs to only one of } L \text{ and } M\\
0 & \text{if } L = M
\end{cases}
\end{equation}
If we set $\lambda = \frac{1}{2}$, then $d_{\mathcal{L}}(\{\epsilon, a, aa, aaa, \dots\}, \{\epsilon, a\}) = \frac{1}{4}$ and $d_{\mathcal{L}}(\{\epsilon,a,aa,aaa, \dots\}, \emptyset)=1$, which allows to formally state that the behaviour of the middle automaton is a better approximation of the left one, rather than the right one. Observe, that we excluded $\lambda=0$ and $\lambda=1$, as in both cases $d_{\mathcal{L}}$ would become a pseudometric setting all languages to be at distance zero or one. Automata in the example above correspond to the regular expressions $a^*$, $a + 1$ and $0$ respectively. In order to determine the distance between arbitrary regular expressions $e$ and $f$ one would have to construct corresponding deterministic finite automata and calculate (or approximate) the distance between their languages. Instead, as a main contribution of this paper, we present a sound and complete quantitative inference system for reasoning about the shortest-distinguishing-word distance of languages denoted by regular expressions in question. Formally speaking, if $\llbracket - \rrbracket : \Exp \to \pset (A^*)$ is a function taking regular expressions to their languages, then our inference system satisfies the following:
$$
	\vdash e \equiv_\e f \iff d_{\mathcal{L}}(\llbracket e \rrbracket, \llbracket f \rrbracket) \leq \e
$$ 
Although much of our development is grounded in category theory and coalgebra, we spell out all the definitions and results concretely, without the need for specialised language. We organise the paper as follows:
\begin{itemize}[topsep=1pt,itemsep=0pt]
	\item In \Cref{sec: preliminaries} we review basic definitions from automata theory and recall the semantics of regular expressions through Brzozowski derivatives~\cite{Brzozowski:1964:Expressions}. Then, in order to talk about distances, we state basic definitions and properties surrounding (pseudo)metric spaces.
	\item In \Cref{sec: behavioural_distance} we instantiate the framework of coalgebraic behavioural metrics~\cite{Baldan:2018:Coalgebraic} to the concrete case of deterministic automata. We recall the abstract results from~\cite{Baldan:2018:Coalgebraic} in simple automata-theoretic terms.
	\item In \Cref{sec: quantitative_equational_theories} we start by recalling the definitions surrounding the quantitative equational theories~\cite{Mardare:2016:Quantitative} from the literature. We then present the axioms of our inference system for the shortest-distinguishing-word distance of regular expressions, give soundness result and provide a discussion about the axioms. The interesting insight is that when relying on quantitative equational theories which contain an infinitary rule capturing the notion of convergence, there is no need for any fixpoint introduction rule. We illustrate this by axiomatically deriving Salomaa's fixpoint rule for regular expressions~\cite{Salomaa:1966:Two}.
	\item The key result of our paper is contained in \Cref{sec: Completeness}, where we prove completeness of our inference system. The heart of the argument relies on showing that the behavioural distance of regular expressions can be approximated from above using Kleene's fixpoint theorem, which can be then mimicked through the means of axiomatic reasoning. This part of the paper makes heavy use of the order-theoretic and Banach space structures carried by the sets of pseudometrics over a given set.
	\item We conclude in \Cref{sec: conclusion}, review related literature, and sketch directions for future work.
\end{itemize}
 
\section{Preliminaries}\label{sec: preliminaries}
We start by recalling basic definitions surrounding deterministic automata, regular expressions and (pseudo)metric spaces from the literature.

\subparagraph*{Deterministic automata.} A deterministic automaton $\mathcal{M}$ with inputs in a finite alphabet $A$ is a pair $(M, \langle o_M , t_M \rangle)$ consisting of a set of states $M$ and a pair of functions $\langle o_M, t_M \rangle$, where $o_M : M \to \{0,1\}$ is the \emph{output} function which determines whether a state $m$ is final ($o_M(m)=1$) or not ($o_M(m)=0$), and $t :M \to M^A$ is the \emph{transition} function, which, given an input letter $a$ determines the next state. If the set $M$ of states is finite, then we call an automaton $\mathcal{M}$ a deterministic finite automaton (\textsf{DFA}). We will frequently write $m_a$ to denote $t_M(m)(a)$ and refer to $m_a$ as the derivative of $m$ for the input $a$. Definition of derivatives can be inductively extended to words $w \in A^{\ast}$, by setting $m_\e = m$ and $m_{aw'} = (m_a)_{w'}$. Note that our definition of deterministic automaton slightly differs from the most common one in the literature, by not explicitly including the initial state. Instead of talking about the language of the automaton, we will talk about the languages of particular states of the automaton. Given a state $m \in M$, we write $L_{\mathcal{M}} (m) \subseteq {A^{\ast}}$ for its language, which is formally defined by $L_{\mathcal{M}}(m) = \{w \in A^\ast \mid o(m_w) = 1\}$. 
Given two deterministic automata $(M, \langle o_M, t_M \rangle)$ and $(N, \langle o_N, t_N \rangle)$, a function $h : M \to N$ is a homomorphism if it preserves outputs and input derivatives, that is $o_N(h(m))=o_M(m)$ and $h(m)_a = h(m_a)$. The set of all languages $\pset(A^{\ast})$ over an alphabet $A$ can be made into a deterministic automaton $(\pset(A^\ast), \langle o_L, t_L\rangle)$, where for $l \in \pset (A^{\ast})$ the output function is given by $o_L(l)=[\epsilon \in l]$ and for all $a \in A$ the input derivative is defined to be $l_a = \{w \mid aw \in l\}$. This automaton is \emph{final}, that is for any other automaton $\mathcal{M} = (M, \langle o_M, t_M \rangle)$ there exists a unique homomorphism from $M$ to $\pset(A^{\ast})$, which is precisely given by the map $L_{\mathcal{M}} : M \to \pset(A^\ast)$ taking each state $m \in S$ to its language. Given a set of states $M' \subseteq M$, we write $\gen{M'}{\mathcal{M}} \subseteq M$ for the smallest set of states reachable from $M'$ through the transition function of the automaton $\mathcal{M}$. Clearly, $(\gen{M'}{\mathcal{M}}, \langle o_M, t_M\rangle)$ is a deterministic automaton.
We will abuse the notation and write $\gen{M'}{\mathcal{M}}$ for $(\gen{M'}{\mathcal{M}}, \langle o_M, t_M\rangle)$.
The canonical inclusion map $\iota : \gen{M'}{\mathcal{M}} \hookrightarrow M$ given by $\iota(m)=m$ for all $m \in \gen{M'}{\mathcal{M}}$ is a homomorphism from $\gen{M'}{\mathcal{M}}$ to $\mathcal{M}$. In the case of singleton and two-element sets of states, we will simplify the notation and write $\gen{m}{\mathcal{M}}$ and $\gen{m, m'}{\mathcal{M}}$.

\subparagraph*{Regular expressions. } We let $e, f$ range over \emph{regular expressions over $A$} generated by the following grammar:
$$e , f \in \Exp ::= 0 \mid 1 \mid a \in A \mid e + f \mid e \seq f \mid e^{\ast}$$
The standard interpretation of regular expressions $\llbracket - \rrbracket : \Exp \to \pset(A^\ast)$ is inductively defined by the following:
$$\llbracket 0 \rrbracket = \emptyset \quad \llbracket 1 \rrbracket = \{\epsilon\} \quad \llbracket a \rrbracket = \{ a \} \quad \llbracket e + f \rrbracket = \llbracket e \rrbracket \cup \llbracket f \rrbracket \quad \llbracket e \seq f \rrbracket = \llbracket e \rrbracket \diamond \llbracket f \rrbracket \quad \llbracket e^\ast \rrbracket = \llbracket e \rrbracket^\ast $$
We write $\epsilon$ for the empty word. Given $L, M \subseteq A^\ast$, we define $L \diamond M = \{lm \mid l \in L, m \in M\}$, where mere juxtaposition denotes concatenation of words. $L^\ast$ denotes the \emph{asterate} of the language $L$ defined as $L^\ast = \bigcup_{i \in \bN} L^i$ with $L^0 = \{\epsilon\}$ and $L^{n + 1} = L \diamond L^{n}$.

\subparagraph*{Brzozowski derivatives. } The famous Kleene's theorem states that the formal languages accepted by \textsf{DFA} are in one-to-one correspondence with formal languages definable by regular expressions. One direction of this theorem involves constructing a \textsf{DFA} for an arbitrary regular expression. The most common way is via Thompson construction, $\epsilon$-transition removal and determinisation. Instead, we recall a direct construction due to Brzozowski~\cite{Brzozowski:1964:Expressions}, in which the set $\Exp$ of regular expressions is equipped with a structure of deterministic automaton $\mathcal{R} = (\Exp, \langle o_{\mathcal{R}}, t_{\mathcal{R}} \rangle)$ through so-called Brzozowski derivatives~\cite{Brzozowski:1964:Expressions}. The output derivative $o_{\mathcal{R}} : \Exp \to \{0, 1\}$ is defined inductively by the following for $a \in A$ and $e,f \in \Exp$:
\begin{gather*}
    o_{\mathcal{R}}(0) = 0 \quad o_{\mathcal{R}}(1) = 1 \quad o_{\mathcal{R}}(a) = 0 \\ o_{\mathcal{R}}(e + f) = o_{\mathcal{R}}(e) \vee o_{\mathcal{R}}(f) \quad
    o_{\mathcal{R}}(e \seq f) = o_{\mathcal{R}}(e) \wedge o_{\mathcal{R}}(f) \quad o_{\mathcal{R}}(e^{\ast}) = 1
\end{gather*}
Similarly, the transition derivative $t_{\mathcal{R}} \in \Exp \to A \to \Exp$ denoted $t_{\mathcal{R}} (e)(a) = (e)_a$ is defined by the following:
\begin{gather*}
    (0)_a = 0 \quad (1)_a = 0 \quad (a')_a = \begin{cases}
        1 & a = a'\\ 0 & a \neq a'  \end{cases}\\
    (e + f)_a = (e)_a + (f)_a \quad
    (e \seq f)_a = (e_a) \seq f + o_{\mathcal{R}}(e) \seq f \quad (e^{\ast}) = (e)_a \seq e^{\ast}
\end{gather*}
The canonical language-assigning homomorphism from $\mathcal{R}$ to $\mathcal{L}$, happens to coincide with the semantics map $\sem{-}$ assigning a language to each regular expression.
\begin{lemma}[{\cite[Theorem~3.1.4]{Silva:2010:Kleene}}]
    For all $e \in \Exp$, $\sem{e} = L_{\mathcal{R}}(e)$
\end{lemma}
Instead of looking at infinite-state automaton defined on the state-space of all regular expressions, we can restrict ourselves to the subautomaton $\gen{e}{\mathcal{R}}$ of $\mathcal{R}$ while obtaining the semantics of $e$.
\begin{restatable}{lemma}{adequacy}\label{lem:adequacy}
    For all $e \in \Exp$, $\sem{e} = L_{\gen{e}{\mathcal{R}}}(e)$
\end{restatable}
Unfortunately, for an arbitrary regular expression $e \in \Exp$, the automaton $\gen{e}{\mathcal{R}}$ is not guaranteed to have a finite set of states. However, simplifying the transition derivatives by removing duplicates in the expressions in the form $e_1 + \dots + e_n$, guarantees a finite number of reachable states from any expression. Formally speaking, let ${\acirel} \subseteq {\Exp \times \Exp}$ be the least congruence relation closed under $ {{(e + f) + g} {~\acirel~} {e + (f+g)}}$ (Associativity), ${e + f} {~\acirel~} {f + e}$ (Commutativity) and ${e} {~\acirel~} {e + e}$ (Idempotence) for all $e,f,g \in \Exp$.
We will write ${\aciq}$ for the quotient of $\Exp$ by the relation $\acirel$ and $[-]_{\acirel} : \Exp \to \aciq$ for the canonical map taking each expression $e \in \Exp$ into its equivalence class $[e]_{\acirel}$ modulo $\acirel$. Because of \cite[Proposition~5.8]{Rutten:2000:Universal}, $\aciq$ can be equipped with a structure of deterministic automaton $\mathcal{Q} = (\aciq, \langle o_{\mathcal{Q}}, t_{\mathcal{Q}} \rangle)$, where for all $e \in \Exp, a \in A$, $o_{\mathcal{Q}}([e]_{\acirel})=o_{\mathcal{R}}(e)$ and $([e]_{\acirel})_a = [e_a]_{\acirel}$, which makes the quotient map $[-]_{\acirel} : \Exp 
\to \aciq$ into an automaton homomorphism from the Brzozowski automaton $\mathcal{R}$ into $\mathcal{Q}$. This automaton enjoys the following property:
\begin{lemma}[{\cite[Theorem~4.3]{Brzozowski:1964:Expressions}}]\label{lem:locally_finite}
    For any $e \in \Exp$, the set $\gen{e}{\mathcal{Q}} \subseteq \aciq$ is finite.
\end{lemma}
Through an identical line of reasoning as before (\Cref{lem:adequacy}), we can show that:
\begin{lemma}
	For all $e \in \Exp$, $L_{\gen{[e]_{\acirel}}{\mathcal{Q}}}([e]_{\acirel}) = \sem{e}$
\end{lemma}
\subparagraph*{(Pseudo)metric spaces. } Let $\top \in \left] 0, \infty \right]$ be a fixed maximal element. A $\top$-bounded \emph{pseudometric} on a set $X$ (equivalently $\top$-\emph{pseudometric} or even just a \emph{pseudometric} if $\top$ is clear from the context) is a function $d : X \times X \to [0, \top]$ satisfying $d(x,x)=0$ (\emph{reflexivity}), $d(x,y)=d(y,x)$ (\emph{symmetry}) and $d(x,z) \leq d(x, y) + d(y,z)$ (\emph{triangle inequality}) for all $x,y,z \in X$. If additionally $d(x,y)=0$ implies $x=y$, $d$ is called a $\top$-\emph{metric}. A \emph{(pseudo)metric} space is a pair $(X,d)$ where $X$ is a set and $d$ is a (pseudo)metric on $X$. Given pseudometric spaces $(X, d_X)$ and $(Y, d_Y)$, we call a map $f : X \to Y$ \emph{nonexpansive}, if for all $x,x' \in X$, $d_Y(f(x),f(x')) \leq d_X(x,x')$ and an \emph{isometry} if $d_Y(f(x),f(x')) = d_X(x,x')$.
A simple example of a pseudometric is the discrete metric which can be defined on any set $X$ as $d_X(x,x)=0$ for all $x \in X$ and $d(x,y)_X=\top$ for $x,y \in X$ such that $x \neq y$.
The set $D_X$ of (pseudo)metrics over some fixed set $X$ can be equipped with a partial order structure given by the pointwise order, i.e. $d \sqsubseteq d' \iff \forall x, x' \in X . d(x,y) \leq d'(x,y)$. 
\begin{lemma}[{\cite[Lemma~3.2]{Baldan:2018:Coalgebraic}}]
    $(D_X, \sqsubseteq )$ is a complete lattice. The join of an arbitrary set of pseudometrics $D \subseteq D_X$ is taken pointwise, ie. $\left(\sup D \right)(x,y) = \sup \{ d(x,y) \mid d \in D\}$ for $x, y \in X$. The meet of $D$ is defined to be $\inf D = \sup \{ d \mid d \in D_X , \forall {d' \in D}, d \sqsubseteq d'\}$.
\end{lemma}
Crucially for our completeness proof, if we are dealing with descending chains, that is sequences $\{d_i\}_{i \in \bN}$, such that $d_i \sqsupseteq d_{i+1}$ for all $i \in \bN$, then we can also calculate infima in the pointwise way\footnote{\Cref{lem:chain_pointwise_inf} is one of the intermediate results used in the proof of~\cite[Lemma~5.6]{Bacci:2018:Bisimilarity} that was communicated to us by the authors of~\cite{Bacci:2018:Bisimilarity}. As this result was excluded in the mentioned paper, we incorporated it along with its proof for the sake of completeness. }.
\begin{restatable}{lemma}{pointwiseinf}\label{lem:chain_pointwise_inf}
    Let $\{d_i\}_{i \in \bN}$ be an infinite descending chain in the lattice $(D_X, \sqsubseteq)$ of pseudometrics over some fixed set $X$. Then $(\inf\{d_i \mid i \in \bN\})(x,y) = \inf \{d_i(x,y) \mid i \in \bN \}$ for any $x,y \in X$.
\end{restatable}
\begin{proof}
    It suffices to argue that $d(x,y) = \inf \{d_i(x,y) \mid i \in \bN \}$ is a pseudometric. For reflexivity, observe that $d(x,x) = \inf \{d_i(x,x) \mid i \in \bN \} = \inf \{ 0 \} = 0 $ for all $x \in X$. For symmetry, we have that $d(x,y) = \inf \{d_i(x,y) \mid i \in \bN \} =  \inf \{d_i(y,x) \mid i \in \bN \}=d(y,x)$ for any $x, y \in X$. The only difficult case is triangle inequality. First, let $i, j \in \bN$ and define $k = \min(i,j)$. Since $d_k \sqsubseteq d_i$ and $d_k \sqsubseteq d_j$, we have that $d_k(x,y) + d_k(y,z) \leq d_i(x,y) + d_j(y,z)$.
    Therefore $\inf \{d_l(x,y) + d_l(y,z) \mid l \in \bN\}$ is a lower bound of $d_i(x,y) + d_j(y,z)$ for any $i, j \in \bN$ and hence it is below the greatest lower bound, that is $\inf \{d_l(x,y) + d_l(y,z) \mid l \in \bN\} \leq \inf \{d_i(x,y) + d_j(y,z) \mid i,j \in \bN\}$. We can use that property to show the following $d(x,y) = \inf \{d_i(z,y) \mid i \in \bN\} \leq \inf \{d_i(x,y) + d_i(y,z) \mid i \in \bN\} \leq \{d_i(x,y) + d_j(y,z) \mid i,j \in \bN\} = \inf \{d_i(x,y) \mid i \in \bN\} + \inf \{d_j(y,z) \mid j \in \bN\} = d(x,y) + d(y,z)$, which completes the proof.
\end{proof}
\noindent
Additionally, the set of pseudometrics can be equipped with a norm. We write $\eR = [-\infty, \infty]$ for the set of extended reals. For any set $X$, the set of functions $\eR^{X \times X }$, which is a superset of $D_X$, can be seen as a Banach space~\cite{Rudin:1990:Functional} (complete normed vector space) by means of the sup-norm $\|d\| = \sup_{x,y \in X} |d(x,y)|$. This structure will implicitly underly some of the claims used as intermediate steps in the proof of completeness.

\section{Behavioural distance}\label{sec: behavioural_distance}
We now instantiate the abstract coalgebraic framework~\cite{Baldan:2018:Coalgebraic} to the case of deterministic automata relying on the lifting described in~\cite[Example~5.33]{Baldan:2018:Coalgebraic}. We concretise the generic results from that paper and spell them in simple automata-theoretic terms.
\subparagraph*{Lifting pseudometrics. } Let $\mathcal{M} = (M, \langle o_M, t_M\rangle)$ be a deterministic automaton. Its one-step observable behaviour (after applying the output and transition derivatives) can be seen as pairs of the type $\{0,1\} \times M^A$, where the first component determines whether the given state is accepting or not and the second one gives successor state for each letter from the input alphabet. Let's say we have two such observations $\langle o_1, f_1 \rangle, \langle o_2, f_2 \rangle \in \{0,1\} \times M^A$. If we had had some notion of a distance defined on the state-space of our automaton, or speaking more formally a $1$-pseudometric $d : M \times M \to [0,1]$, then we can \emph{lift} this notion of distance, to a distance on observations, given by the following:
\[
	d_{\{0,1\} \times M^A}(\langle o_1, f_1 \rangle, \langle o_2, f_2 \rangle) = \max\{d_{2}(o_1, o_2), \lambda \cdot \max_{a \in A} d(f_1(a), f_2(a))\} \qquad \lambda \in \left] 0,1\right[
\]   
The definition above involves $d_2$, the discrete metric on the set $\{0,1\}$. One can observe that $d_{\{0,1\} \times M^A}$ is again a $1$-pseudometric, but this time defined on the set $\{0,1\} \times M^A$ instead. We now move on to showing how one could use this lifting in order to equip a state-space of automaton $\mathcal{M}$ with a sensible notion of behavioural pseudometric.
\subparagraph*{Behavioural pseudometric. } 
If one gave us a $1$-metric $d : M \times M \to [0,1]$ on the state-space of the automaton, we could use our lifting to produce a new pseudometric $\Phi_{\mathcal{M}}(d) : M \times M \to [0,1]$ on the same set, which would calculate a distance between an arbitrary pair of states by first applying the output and transition derivatives to obtain a pair of observations and then by calculating the distance between them using the aforementioned lifting of pseudometric $d$ to the pseudometric defined on the set $\{0,1\} \times M^A$. Formally speaking, define a map $\Phi_{\mathcal{M}} : D_M \to D_M$ on the lattice of $1$-pseudometrics on $M$ given by the following:
$$
\Phi_{\mathcal{M}}(d)(m,m') = \max\{d_{2}(o_M(m), o_M(m')), \lambda \cdot \max_{a \in A} d(m_a, m'_a)\} \qquad \lambda \in \left] 0,1\right[
$$
The construction above only tells us how to construct new pseudometrics on the state-space of the automaton out of existing ones, but does not give one to start with. It turns out, that the map $\Phi_{\mathcal{M}}$ is a monotone mapping on the lattice of $1$-psuedometrics on the set $M$~\cite[Lemma~6.1]{Baldan:2018:Coalgebraic}. Because of that, one can use the Knaster-Tarski fixpoint theorem~\cite{Tarski:1955:Lattice} and construct its least fixed point, explicitly given by $d_{\mathcal{M}} = \inf \{d \mid d \in D_M \wedge \Phi_{\mathcal {M}}(d) \sqsubseteq d \}$. Pseudometrics, which are fixpoints of $\Phi_{\mathcal{M}}$ intuitively interact well with the automaton structure, as they satisfy the property that the distance between two states is the same as the distance between their observable behaviour calculated using the lifting. Taking the least such pseudometric satisfies several desirable properties~\cite{Baldan:2018:Coalgebraic} and thus we will call $d_{\mathcal{M}}$ a behavioural pseudometric on the automaton $\mathcal{M}$. First of all, preserving automaton transitions also preserves behavioural distances.
\begin{restatable}{proposition}{homomorphismsisometries}\label{prop:homomorphisms_isometries}
    Let $\mathcal{M} = (M, \langle o_M, t_M \rangle)$ and $\mathcal{N} = (N, \langle o_N, t_N \rangle)$ be deterministic automata. If $h : M \to N$ is a homomorphism, then it is also an isometric mapping between pseudometric spaces $(M, d_{\mathcal{M}})$ and $(N, d_{\mathcal{N}})$.
\end{restatable}
If we look at the final automaton on the set of all formal languages over a fixed finite alphabet $A$, then one can easily verify that the behavioural distance given by the least fixpoint construction precisely corresponds to the \Cref{eq: shortest_distinguishing_word} defining the shortest-distinguishing-word distance we stated in \Cref{sec: introduction}. In general, states of an arbitrary deterministic automaton characterised by the behavioural pseudometric to be in distance zero are language equivalent. When we look at $d_{\mathcal{L}}$ defined on the states of the final automaton, whose state-space consists of formal languages, then the language equivalence corresponds to the equality of states. In other words $d_{\mathcal{L}}$ becomes a metric space.
\begin{restatable}{lemma}{behaviouralpseudometric}
    Let $\mathcal{M} = (M, \langle o_M, t_M\rangle)$ be an arbitrary deterministic automaton and let $\mathcal{L} = (\pset (A^{\ast}), \langle o_L , t_L \rangle)$ be a deterministic automaton structure on the set of all languages over an alphabet $A$.
    \begin{enumerate}[topsep=1pt,itemsep=0pt]
        \item $(\pset (A^\ast), d_{\mathcal{L}})$ is a metric space.
        \item For any $m, m' \in M$, $d_\mathcal{M}(m,m') = 0 \iff L_{\mathcal{M}}(m) = L_{\mathcal{M}}(m')$.
    \end{enumerate} 
\end{restatable}

\section{Quantitative Axiomatisation}\label{sec: quantitative_equational_theories}
In order to provide a quantitative inference system for reasoning about the behavioural distance of languages denoted by regular expressions, we first recall the definition of quantitative equational theories from the existing literature~\cite{Mardare:2016:Quantitative,Bacci:2018:Bisimilarity} following the notational conventions from~\cite{Bacci:2018:Bisimilarity}. We then present our axiomatisation and demonstrate its soundness. The interesting thing about our axiomatisation is the lack of any fixpoint introduction rule. We show that in the case of quantitative analogue of equational logic~\cite{Mardare:2016:Quantitative} containing the infinitary rule capturing the notion of convergence, we can use our axioms to derive Salomaa's fixpoint rule from his axiomatisation of language equivalence of regular expressions~\cite{Salomaa:1966:Two}.

\subparagraph*{Quantitative equational theories.} 
Let $\Sigma$ be an algebraic signature (in the sense of universal algebra~\cite{Burris:1981:Course}) consisting of operation symbols $f_n \in \Sigma$ of arity $n \in \bN$. If we write $X$ for the countable set of \emph{metavariables}, then $\TT[\Sigma]{X}$ denotes a set of freely generated terms over $X$ built from the signature $\Sigma$. As a notational convention, we will use letters $t,s,u, \ldots \in \TT[\Sigma]{X}$ to denote terms. 
By a \emph{substitution} we mean a function of the type $\sigma\colon X \to \TT[\Sigma]{X}$ allowing to replace metavariables with terms. Each substitution can be inductively extended to terms in a unique way by setting $\sigma(f(t_1, \dots, t_n)) = f(\sigma(t_1), \dots, \sigma(t_n))$ for each operation symbol $f_n \in \Sigma$ from the signature. We will write $\Sub[\Sigma]$ for the set of all substitutions. Given two terms $t,s \in \TT[\Sigma]{X}$ and a nonnegative rational number $\e \in \bQ$ denoting the distance between the terms, we call $t \equiv_\e s$ a \emph{quantitative equation (of type $\Sigma$)}. Notation-wise, we will write $\E[\Sigma]$ to denote the set of all quantitative equations (of type $\Sigma$) and we will use the capital Greek letters $\Gamma, \Theta, \ldots \subseteq \E[\Sigma]$ to denote the subsets of $\E[\Sigma]$. By a \emph{deducibility relation} we mean a binary relation denoted  ${\vdash} \subseteq \pset ({\E[\Sigma]}) \times \E[\Sigma]$.  Similarly, to the classical equational logic, we will use the following notational shorthands $\Gamma \vdash t \equiv_\e s \iff (\Gamma,t \equiv_\e s) \in {\vdash}$ and $\vdash t \equiv_\e s \iff \emptyset \vdash t \equiv_\e s$. Furthermore, following the usual notational conventions, we will write $\Gamma \vdash \Theta$ as a shorthand for the situation when $\Gamma \vdash t \equiv_\e s$ holds for all $t \equiv_\e s \in \Theta$. To call $\vdash$ a \emph{quantitative deduction 
system (of type $\Sigma$)} it needs to satisfy the following rules of inference: 
\begin{align*} 
(\Refl) \quad 
& \vdash t \equiv_0 t \,, \\
(\Symm) \quad 
& \{t\equiv_\e s\} \vdash s\equiv_\e t \,, \\
(\Triang) \quad 
& \{t \equiv_\e u, u \equiv_{\e'} s \} \vdash t \equiv_{\e+\e'} s \,, \\
(\Max) \quad 
& \{t\equiv_\e s\} \vdash t\equiv_{\e+\e'}s \,, \text{ for all $\e'>0$} \,, \\ 
(\Cont) \quad 
& \{t\equiv_{\e'}s\mid \e'>\e\} \vdash t\equiv_\e s \,, \\
(\Nexp) \quad
& \{t_1\equiv_\e s_1,\ldots,t_n \equiv_\e s_n\} \vdash f(t_1,\dots, t_n) \equiv_\e f(s_1,\dots, s_n) \,, 
\text{ for all $f_n \in \Sigma$} \,, \\
(\Subst) \quad
& \text{If $\Gamma \vdash t \equiv_\e s$, then $\sigma(\Gamma) \vdash \sigma(t) \equiv_\e \sigma(s)$, 
for all $\sigma \in \Sub[\Sigma]$} \,, \\
(\Cut) \quad 
& \text{If $\Gamma \vdash \Theta$ and $\Theta \vdash t \equiv_\e s$, then $\Gamma \vdash t \equiv_\e s$} \,, \\
(\Assum) \quad
& \text{If $t \equiv_\e s \in\Gamma$, then $\Gamma \vdash t \equiv_\e s$} \,.
\end{align*}
where $\sigma(\Gamma) = \set{\sigma(t) \equiv_\e \sigma(s)}{ t \equiv_\e s \in \Gamma}$.
Finally, by a \emph{quantitative equational theory} we mean a set $\U$ of universally quantified \emph{quantitative inferences} 
$
\{t_1 \equiv_{\e_1} s_1, \dots, t_n \equiv_{\e_n} s_n\} \vdash t \equiv_\e s \,,
$ with \emph{finitely many premises}, closed under $\vdash$-derivability.

\subparagraph*{Quantitative algebras.} Quantitative equational theories lie on the syntactic part of the picture. On the semantic side, we have their models called \emph{quantitative algebras}, defined as follows. 

\begin{definition}[{\cite[Definition~3.1]{Mardare:2016:Quantitative}}]
    A quantitative algebra is a tuple $\A = (A, \Sigma^{\A}, d^{\A})$, such that $(A, \Sigma^{\A})$ is an algebra for the signature $\Sigma$ and $(A, d^{\A})$ is an $\infty$-pseudometric such that for all operation symbols $f_n \in \Sigma$, for all $1 \leq i \leq n$, $a_i, b_i \in A$, $d^{\A}(a_i,b_i)\leq \e$ implies $d^{\A}(f^{\A}(a_1, \dots, a_n), f^{\A}(b_1, \dots, b_n)) \leq \e$.
\end{definition}

Consider a quantitative algebra $\A = (A,\Sigma^\A,d^\A)$. Given an assignment  $\iota \colon X \to A$ of meta-variables from $X$ to elements of carrier $A$, one can inductively extend it to $\Sigma$-terms $t \in \TT[\Sigma]{X}$ in a unique way. We will abuse the notation and just write $\iota(t)$ for the interpretation of the term $t$ in quantitative algebra $\A$. We will say that $\A$ \emph{satisfies} the quantitative inference $\Gamma \vdash t \equiv_\e s$, written $\Gamma \models_\A t \equiv_\e s$, if for any assignment of the meta-variables $\iota \colon X \to A$ it is the case that for all $t' \equiv_{\e'} s' \in \Gamma$ we have that $d^\A(\iota(t'),\iota(s')) \leq \e'$ implies $d^\A(\iota(t),\iota(s)) \leq \e $. Finally, we say that a quantitative algebra $\A$ \emph{satisfies} (or is a \emph{model} of) the quantitative theory $\U[]$, 
if whenever $\Gamma \vdash t \equiv_\e s \in \U[]$, then $\Gamma \models_\A t \equiv_\e s$. 

\subparagraph*{Quantitative algebra of regular expressions.} From now on, let's focus on the signature $\Sigma = \{0_0, 1_0, +_2, \seq_2, {(-)^\ast}_1\} \cup \{a_0 \mid a \in A\}$, where $A$ is a finite alphabet. This signature consists of all operations of regular expressions. We can easily interpret all those operations in the set $\Exp$ of all regular expressions, using trivial interpretation functions eg. $+^{\B}(e,f) = e + f$, which interpret the operations by simply constructing the appropriate terms. Formally speaking, we can do this because the set $\Exp$ is the carrier of initial algebra~\cite{Burris:1981:Course} (free algebra over the empty set of generators) for the signature $\Sigma$. 

To make this algebra into a quantitative algebra, we first equip the set $\Exp$ with a $\infty$-pseudometric, given by $
d^{\B}(e,f) = d_\mathcal{L}(\llbracket e \rrbracket, \llbracket f \rrbracket)$ for all $e, f \in \Exp$. 
Recall that $d_{\mathcal{L}}$ used in the definition above is a behavioural pseudometric on the final deterministic automaton carried by the set $\pset (A^{\ast})$ of all formal languages over an alphabet $A$. In other words, we define the distance between arbitrary expressions $e$ and $f$ to be the distance between formal languages $\sem{e}$ and $\sem{f}$ calculated through the shortest-distinguishing-word metric. It turns out, that in such a situation all the interpretation functions of $\Sigma$-algebra structure on $\Exp$ are non-expansive with respect to to the pseudometric defined above. In other words, we have that: 

\begin{restatable}{lemma}{quantitativealgebra}\label{lem:quantitative_algebra}
    $\B = (\Exp, \Sigma^{\B}, d^{\B})$ is a quantitative algebra.
\end{restatable}

\subparagraph*{Axiomatisation. } In order to talk about the quantitative algebra $\B$ of the behavioural distance of regular expressions in an axiomatic way, we introduce the quantitative equational theory \textsf{REG} (\Cref{fig:axioms}).
\begin{figure}[h]
	
{
\small
\begin{tabular}{l@{\quad}l}
\(
\begin{array}{ll}
&\textbf{Nondeterministic choice}\\
(\mathsf{SL1})\;\; 
& \vdash e + e \equiv_0 e \,, \\
(\mathsf{SL2})\;\; 
& \vdash e + f \equiv_0 f + e \,, \\
(\mathsf{SL3})\;\; 
& \vdash (e + f) + g\equiv_0 e + (f + g) \,, \\
(\mathsf{SL4})\;\; 
& \vdash e + 0\equiv_0 e \,, \\
(\mathsf{SL5})\;\; & \{ e \equiv_\e g , f \equiv_{\e'} h\} \\ & \quad\vdash e + f \equiv_{\max(\e, \e')} g + h \,, \\\\\\
&\textbf{Loops} \\
(\mathsf{Unroll})\;\; & \vdash e^\ast \equiv_0 e \seq e^{\ast} + 1 \,, \\
(\mathsf{Tight})\;\; & \vdash (e + 1)^\ast \equiv_0 e^\ast \,, \\[1.2ex]
\end{array}
\)&
\(
\begin{array}{ll}
&\textbf{Sequential composition}\\
	(\mathsf{1S})\;\; 
& \vdash 1 \seq e \equiv_0 e \,, \\
(\mathsf{S})\;\; 
& \vdash e \seq (f \seq g) \equiv_0 (e \seq f) \seq g \,, \\
(\mathsf{S1})\;\; 
& \vdash e \seq 1 \equiv_0 e \,, \\
(\mathsf{0S})\;\; 
& \vdash 0 \seq e \equiv_0 0 \,, \\
(\mathsf{S0})\;\; 
& \vdash e \seq 0 \equiv_0 0 \,, \\
(\mathsf{D1})\;\; 
& \vdash e \seq (f + g) \equiv_0 e \seq f + e \seq g \,, \\
(\mathsf{D2})\;\; 
& \vdash (e + f) \seq g \equiv_0 e \seq g + f \seq g \,, \\\\
&\textbf{Behavioural pseudometric}\\
(\Top)\;\; & \vdash e \equiv_1 f \,, \\
(\dPref)\;\; & \{ e \equiv_\e f \} \vdash a \seq e \equiv_{\e'} a \seq f \,, \text{for $\e'\geq \lambda \cdot \e$} \\[1.2ex]
\end{array}
\)
\end{tabular}
}	
\caption{Axioms of the quantitative equational theory \textsf{REG} for $e,f,g \in \Exp$ and $a \in A$.}
\label{fig:axioms}
\end{figure}
The first group of axioms capture properties of the nondeterministic choice operator $+$ \textsf{(SL1-SL5)}. The first four axioms \textsf{(SL1-SL4)} are the usual laws of semilattices with bottom element $0$. (\textsf{SL5}) is a quantitative axiom allowing one to reason about distances between sums of expressions in terms of distances between expressions being summed. Moreover, \textsf{(SL1-SL5)} are axioms of so-called \emph{Quantitative Semilattices with zero}, which have been shown to axiomatise the Hausdorff metric~\cite{Mardare:2016:Quantitative}. The sequencing axioms \textsf{(1S), (S1), (S)} state that the set $\Exp$ of regular expressions has the structure of a monoid (with neutral element $1$) with absorbent element $0$ \textsf{(0S), (S0)}. Additionally, \textsf{(D1-D2)} talk about interaction of the nondeterministic choice operator $+$ with sequential composition. The loop axioms \textsf{(Unroll)} and \textsf{(Tight)} are directly inherited from Salomaa's axiomatisation of language equivalence of regular expressions~\cite{Salomaa:1966:Two}. \textsf{(Unroll)} axiom associates loops with their intuitive behaviour of choosing, at each step, between successful termination and executing the loop body once. \textsf{(Tight)} states that the loop whose body might instantly terminate, causing the next loop iteration to be executed immediately is provably equivalent to a different loop, whose body does not contain immediate termination. 
The last remaining group are behavioural pseudometric axioms. \textsf{(Top)} states that any two expressions are at most in distance one from each other. Finally, \textsf{(\dPref)} captures the fact that prepending the same letter to arbitrary expressions shrinks the distance between them by the factor of $\lambda \in ]0,1[$ (used in the definition of $d^{\B}$). This axiom is adapted from the axiomatisation of discounted probabilistic bisimilarity distance~\cite{Bacci:2018:Bisimilarity}.

Through a simple induction on the length of derivation, one can verify that indeed $\B$ is a model of the quantitative theory \textsf{REG}. 
\begin{restatable}{theorem}{soundness}{(Soundness)}
    The quantitative algebra $\B = (\Exp, \Sigma^{\B}, d^\B)$ is a model of the quantitative theory $\mathsf{REG}$. In other words, for any $e, f \in \Exp$ and $\e \in \bQ$, if $\Gamma \vdash e \equiv_\e f \in \mathsf{REG}$, then $\Gamma \models_\B e \equiv_\e f$
\end{restatable}
\begin{proof}
By the structural induction on the judgement $\Gamma \vdash e \equiv_{\e} f \in \mathsf{REG}$. $(\Subst)$, $(\Cut)$ and $(\Assum)$ deduction rules from classical logic hold immediately. The soundness of $(\Refl)$, $(\Symm)$, $(\Triang)$, $(\Cont)$ and $(\Max)$ follows from the fact that $d^{\B}$ is a pseudometric. $(\Nexp)$ follows from the fact that interpretations of symbols from the algebraic signature are nonexpansive (\cref{lem:quantitative_algebra}). Recall that $d^{\B} = d_{\mathcal{L}} \circ (\sem{-} \times \sem {-})$. The soundness of $(\Top)$ follows from the fact that $d_{\mathcal{L}}$ is a 1-pseudometric. Additionally, for all axioms in the form $\vdash e \equiv_0 f$ it suffices to show that $\sem{e} = \sem{f}$. $(\mathsf{SL1})$, $(\mathsf{SL2})$, $(\mathsf{SL3})$, $(\mathsf{SL4})$, $(\mathsf{1S})$, $(\mathsf{S})$, $(\mathsf{S1})$, $(\mathsf{0S})$, $(\mathsf{S0})$, $(\mathsf{D1})$, $(\mathsf{D2})$, $(\mathsf{Unroll})$ and $(\mathsf{Tight})$ are taken from Salomaa's axiomatisation of language equivalence of regular expressions~\cite{Salomaa:1966:Two} and thus both sides of those equations denote the same formal languages~\cite[Theorem~5.2]{Wagemaker:2019:Completeness}. For $(\dPref)$ assume that the premise is satisfied in the model, that is $d_{\mathcal{L}}(\sem{e},\sem{f}) \leq \e$. Let $\e' \geq \lambda \cdot \e$. We show the following:
\begin{align*}
    d^{\B}(a \seq e, a \seq f) &= d_{\mathcal{L}}(\sem{a \seq e}, \sem{a \seq f}) \tag{Def. of $d^{\B}$}\\
    & = \Phi_{\mathcal{L}}(d_{L})(\sem{a \seq e}, \sem{a \seq f}) \tag{$d_L$ is a fixpoint of $\Phi_{\mathcal{L}}$}\\
    &= \max\{d_{2}(o_{\mathcal{L}}(a \seq e), o_{\mathcal{L}}(a \seq f)), \lambda \cdot \max_{a' \in A} d_{\mathcal{L}}(\sem{a \seq e}_{a'}, \sem{a \seq f}_{a'} )\}\\
    &= \lambda \cdot d_{\mathcal{L}}(\sem{e}, \sem{f}) \tag{Def. of final automaton} \\
    & \leq \lambda \cdot \epsilon \leq \e' \tag{Assumptions} \\
\end{align*}
Finally, $(\mathsf{SL5})$ is derivable from other axioms\footnote{We included $(\mathsf{SL5})$ as an axiom to highlight the similarity of our inference system with axiomatisations of language equivalence of regular expressions~\cite{Salomaa:1966:Two,Kozen:1994:Completeness} containing the axioms of semilattices with bottom. In the previous work~\cite{Mardare:2016:Quantitative}, $(\mathsf{SL1-SL5})$ are precisely the axioms of \emph{Quantitative Semilattices with zero} axiomatising the Hausdorff distance.}. 
If $\e = \max(\e, \e')$ then $\{e \equiv_\epsilon g \} \vdash e \equiv_{\max(\e, \e')} g$ holds by $(\Assum)$. If $\epsilon < \max(\e, \e')$, then we can derive the quantitative judgement above using $(\Max)$. By a similar line of reasoning, we can show that $\{f \equiv_{\epsilon'} h \} \vdash f \equiv_{\max(\e, \e')} h$. Finally, using $(\Cut)$ and $(\Nexp)$, we can show that $\{e \equiv_{\e} g, f \equiv_{\e'} h\} \vdash e + f \equiv_{\max(\e,\e')} g + h$ as desired.
\end{proof}
We now revisit the example from \Cref{sec: introduction}. Recall that states marked as initial of the left and middle automata can be respectively represented as $a^*$ and $a + 1$. The shortest word distinguishing languages representing those expressions is $aa$. If we fix $\lambda = \frac{1}{2}$, then $d_{\mathcal{L}}(\sem{a^*}, \sem{a+ 1}) = \frac{1}{4}=\left(\frac{1}{2}\right)^{|aa|}$. We can derive this distance through the means of axiomatic reasoning using the quantitative equational theory \textsf{REG} in the following way: 
\begin{example}
 	\begin{align*}
		\vdash a^* &\equiv_1 0 \tag{$\Top$} \\
		\vdash a \seq a^{*} &\equiv_{\frac{1}{2}} a \seq 0 \tag{$\dPref$}\\
		\vdash a \seq a^* + 1&\equiv_{\frac{1}{2}} a \seq 0 + 1  \tag{$\vdash 1 \equiv_0 1$ and \textsf{SL5}}\\
		\vdash a^{*} &\equiv_{\frac{1}{2}} 1 \tag{$\Triang$, \textsf{Unroll}, \textsf{S0} and \textsf{SL4}}\\
		\vdash a \seq a^* &\equiv_{\frac{1}{4}} a \seq 1 \tag{$\dPref$}\\
		\vdash a \seq a^* + 1 &\equiv_{\frac{1}{4}} a \seq 1 + 1  \tag{$\vdash 1 \equiv_0 1$ and \textsf{SL5}}\\
		\vdash a^* &\equiv_{\frac{1}{4}} a + 1 \tag{$\Triang$, \textsf{Unroll} and \textsf{S1}}\\
	\end{align*}
\end{example}
\subparagraph*{(The lack of) the fixpoint axiom.} Traditionally, completeness of inference systems for behavioural equivalence of languages of expressions featuring recursive constructs such as Kleene star or $\mu$-recursion~\cite{Milner:1984:Complete} rely crucially on fixpoint introduction rules. Those allow showing that an expression is provably equivalent to a looping construct if it exhibits some form of self-similarity, typically subject to productivity constraints. As an illustration, Salomaa's axiomatisation of language equivalence of regular expressions incorporates the following inference rule:
\begin{equation}\label{eq: salomaa}
	\inferrule{g \equiv e \seq g + f \qquad \epsilon \notin \sem{e}}{g \equiv e^* \seq f}
\end{equation}
The side condition on the right states that the loop body is \emph{productive}, that is a deterministic automaton corresponding to an expression $e$ cannot immediately reach acceptance without performing any transitions. This is simply equivalent to the language $\sem{e}$ not containing the empty word. It would be reasonable for one to expect \textsf{REG} to contain a similar rule to be complete, especially since it should be able to prove language equivalence of regular expressions (by proving that they are in distance zero from each other). Furthermore, all axioms of Salomaa except \Cref{eq: salomaa} are contained in \textsf{REG} as rules for distance zero.

It turns out that in the presence of the infinitary continuity \textsf{(Cont)} rule of quantitative deduction systems and the \textsf{(\dPref)} axiom of \textsf{REG}, the Salomaa's inference rule (\Cref{eq: salomaa}) becomes a derivable fact for distance zero. First of all, one can show that \textsf{(\dPref)} can be generalised from prepending single letters to prepending any regular expression satisfying the side condition from \Cref{eq: salomaa}.

\begin{restatable}{lemma}{generalisedpref}\label{lem:generalised_pref}
	Let $e,f,g \in \Exp$, such that $\epsilon \notin \sem{e}$. Then, 	$
		 \{ f \equiv_\e g \} \vdash e \seq f \equiv_{\e'} e \seq g 
	$ is derivable  using the axioms of \textsf{REG}
 for all $\e'\geq \lambda \cdot \e$.
\end{restatable}
With the above lemma in hand, one can inductively show that if $g \equiv_0 e \seq g + f$ and $\epsilon \notin \sem{e}$, then $g$ gets arbitrarily close to $e^* \seq f$. Intuitively, the more we unroll the loop in $e^* \seq f$ using \textsf{(Unroll)} and the more we unroll the definition of $g$, then the closer both expressions become.
\begin{restatable}{lemma}{starlemma}\label{lem:star_lemma}
Let $e,f,g \in \Exp$, such that $\epsilon \notin \sem{e}$ and let $n \in \bN$. Then, $ \{g \equiv_0 e \seq g  + f\} \vdash g \equiv_{\e} e^{*} \seq f$ is derivable using the axioms of $\mathsf{REG}$ for all $\e \geq \lambda^n $.
\end{restatable}
Having the result above, we can now use the infinitary \textsf{(Cont)} rule capturing the limiting property of decreasing chain of overapproximations to the distance and show the derivability of Salomaa's inference rule.
\begin{lemma}
	Let $e,f,g \in \Exp$, such that $\epsilon \notin \sem{e}$. Then, $\{g \equiv_0 e \seq g + f\} \vdash g \equiv_{0} e^{*} \seq f$ is derivable using the axioms of \textsf{REG}.
	
\end{lemma}
\begin{proof}
To deduce that $\vdash g \equiv_0 e^{*} \seq f$ using \textsf{(Cont)} it suffices to show that $\vdash g \equiv_{\e} e^{*} \seq f$ for all $\e > 0$.  To do so, pick an arbitrary $\e > 0$ and let $N = \lceil \log_{\lambda} \e \rceil$. Observe that $\lambda^N = \lambda^{\lceil \log_{\lambda} \e \rceil} \leq \lambda^{\log_{\lambda} \e} = \e$. Because of \Cref{lem:star_lemma} we have that $\vdash g \equiv_{\e} e^{*} \seq f$, which completes the proof.
\end{proof}
\section{Completeness}\label{sec: Completeness}
We now move on to the central result of this paper, which is the completeness of \textsf{REG} with respect to the shortest-distinguishing-word metric on languages denoting regular expressions. We use the strategy from the proof of completeness of quantitative axiomatisation of probabilistic bisimilarity distance~\cite{Bacci:2018:Bisimilarity}. It turns out that the results from~\cite{Bacci:2018:Bisimilarity} rely on properties that are not unique to the Kantorovich/Wassertstein lifting and can be also established for instances of the abstract coalgebraic framework~\cite{Baldan:2018:Coalgebraic}.

The heart of our argument relies on the fact that the distance between languages denoting regular expressions can be calculated in a simpler way than applying the Knaster-Tarski fixpoint theorem while looking at the infinite-state final automaton of all formal languages over some fixed alphabet. In particular, regular expressions denote the behaviour of finite-state deterministic automata. Since automata homomorphisms are non-expansive mappings, the distance between languages $\sem{e}$ and $\sem{f}$ of some arbitrary regular expressions $e, f \in \Exp$ is the same as the distance between states in some \textsf{DFA} whose languages corresponds to $\sem{e}$ and $\sem{f}$. To be precise, we will look at the finite subautomaton $\gen{[e]_{\acirel}, [f]_{\acirel}}{\mathcal{Q}}$ of the $\acirel$ quotient of the Brzozowski automaton. The reason we care about deterministic finite automata is that it turns out that one can calculate the behavioural distance between two states through iterative approximation from above, which can be also derived axiomatically using the \textsf{(Cont)} rule of quantitative deduction systems. We start by showing how this simplification works and then we move on to establishing completeness.

\subparagraph*{Behavioural distance on finite-state automata. } Consider a deterministic automaton $\mathcal{M} = (M, \langle o_M, t_M\rangle)$. The least fixpoint of a monotone endomap $\Phi_{\mathcal{M}} : D_{M} \to D_{M}$ on the complete lattice of $1$-pseudometrics on the set $M$ results in $d_{\mathcal{M}}$, which is a behavioural pseudometric on the states of the automaton $\mathcal{M}$.
It is noteworthy that $\Phi_{\mathcal{M}}$ exhibits two generic properties. Firstly, $\Phi_{\mathcal{M}}$ behaves well within the Banach space structure defined by the supremum norm.
\begin{restatable}{lemma}{operatornonexpansive}\label{lem:operator_nonexpansive}
    $\Phi_{\mathcal{M}} : D_M \to D_M$ is nonexpansive with respect to the supremum norm. In other words, for all $d, d' \in D_M$ we have that
    $\|\Phi_{\mathcal{M}}(d') - \Phi_{\mathcal{M}}(d) \| \leq \|d' - d \|$.
\end{restatable}
\begin{proof}
    We can safely assume that $d \sqsubseteq d'$, as other case will be symmetric. It sufices to show that for all $m, m' \in M$, $\Phi_{\mathcal{M}}(d')(m,m') - \Phi_{\mathcal{M}}(d)(m,m') \leq \|d' - d \|$. First, let's consider the case when $o_M(m) \neq o_M(m')$ and hence $d_2(m,m')=1$. In such a scenario, it holds that $\Phi_{\mathcal{M}}(d')(m,m') - \Phi_{\mathcal{M}}(d)(m,m') = 0 \leq \|d'-d \|$. From now on, we will assume that $o_M(m) = o_M(m)$ and hence $d_2(m,m')=0$. We have the following
    \begin{align*}
        \Phi_{\mathcal{M}}(d')(m,m') - \Phi_{\mathcal{M}}(d)(m,m') &= \lambda \cdot \max_{a \in A} d'(m_a, m'_a) - \lambda \cdot \max_{a \in A} d(m_a, m'_a) \\
        &= \lambda \cdot \left( \max_{a \in A} d'(m_a, m'_a) - \max_{a \in A} d(m_a,  m'_a) \right) \\
        &\leq \lambda \cdot \left(\max_{a \in A} \{d'(m_a, m'_a) - d(m_a, m'_a) \}\right) \\  
        &\leq \lambda \cdot \sup_{n,n' \in M} \{ d'(n,n') - d(n,n')\} \\
        &= \lambda \cdot \|d' - d \| \leq \| d' - d \|
    \end{align*}
\end{proof}
Secondly, it turns out that $\Phi_{\mathcal{M}}$ has only one fixpoint. This means that if we want to calculate $d_{\mathcal{M}}$ it suffices to look at any fixpoint of $\Phi_{\mathcal{M}}$. This will enable a simpler characterisation, than the one given by the Knaster-Tarski fixpoint theorem.
\begin{restatable}{lemma}{uniquefp}\label{lem: uniquefp}
    $\Phi_{\mathcal{M}}$ has a unique fixed point.
\end{restatable}
\begin{proof}
    Let $d, d' \in D_M$ be two fixed points of $\Phi_{\mathcal{M}}$, that is $\Phi_{\mathcal{M}}(d)=d$ and $\Phi_{\mathcal{M}}(d')=d'$. We can safely assume that $d \sqsubseteq d'$, as the other case is symmetric. We wish to show that $d = d'$ and to do so we will use proof by contradiction. 
    
    Assume that $d \neq d'$, and hence there exist $m,m' \in M$, such that $d(m,m')<d'(m,m')$ and $\|d' - d \| = d'(m,m') - d(m,m') \neq 0 $. First, consider the case when $o_M(m) \neq o_M(m')$. In such a case both $d(m,m')$ and $d(m,m')$ are equal to $1$ and hence $\|d' - d \| = 0$, which leads to contradiction. From now, we can safely assume that $o_M(m)=o_M(m')$. Through an identical line of reasoning to the proof of \cref{lem:operator_nonexpansive}, we can show that $\|\Phi_{\mathcal{M}}(d') - \Phi_{\mathcal{M}}(d)\| \leq \lambda \cdot \|d' - d \|$. Since both $d$ and $d'$ are fixed points, this would mean that $\|d' - d \| \leq \lambda \cdot \| d' - d \|$. Since $\lambda \in \left]0, 1 \right[$, this would imply that $\|d' - d\| = 0$ leading again to contradiction.
\end{proof}
In particular, we will rely on the characterisation given by the Kleene fixpoint theorem~\cite[Theorem~2.8.5]{Sangiorgi:2011:Coinduction}, which allows to obtain the greatest fixpoint of an endofunction on the lattice as the infimum of the decreasing sequence of finer approximations obtained by repeatedly applying the function to the top element of the lattice.

\begin{theorem}[Kleene fixpoint theorem]\label{thm: kleene}
	Let $(X, \sqsubseteq)$ be a complete lattice with a top element $\top$ and $f : X \to X$ an endofunction that is $\omega$-cocontinuous or in other words for any decreasing chain $\{x_i\}_{i \in \bN}$ it holds that $\inf_{i \in \bN} \{f(x_i)\} = f \left( \inf_{i \in \bN} x_i \right)$.  Then, $f$ possesses a greatest fixpoint, given by 
	$\operatorname{gfp}(f) = \inf_{i \in \bN }\{f^{(i)}(\top)\}$
	where $f^{(n)}$ denotes $n$-fold self-composition of $f$ given inductively by $f^{(0)}(x)=x$ and $f^{(n+1)}(x) = f^{(n+1)}(f(x))$ for all $x \in X$.
\end{theorem} 
The theorem above requires the endomap to be $\omega$-cocontinuous. Luckily, it is the case for $\Phi_{\mathcal{M}}$ if we restrict our attention to \textsf{DFA}. To show that, we directly follow the line of reasoning from~\cite[Lemma~5.6]{Bacci:2018:Bisimilarity} generalising the similar line of reasoning for $\omega$-continuity from \cite[Theorem~1]{Breugel:2012:On}. First, using  \Cref{lem:chain_pointwise_inf} we show that decreasing chains of pseudometrics over a finite set converge to their infimum. That result is a minor re-adaptation of~\cite[Theorem~1]{Breugel:2012:On} implicitly used in \cite[Lemma~5.6]{Bacci:2018:Bisimilarity}.
\begin{restatable}{lemma}{chainconv}\label{lem:chain_conv_to_inf}
    Let $\{d_i\}_{i \in \bN}$ be an infinite descending chain in the lattice $(D_X, \sqsubseteq)$, where $X$ is a finite set. The sequence $\{d_i\}_{i \in \bN}$ converges (in the sense of convergence in the Banach space) to $d(x,y) = \inf_{i \in \bN} d_i(x,y)$.
\end{restatable}
\begin{proof}
     Let $\e > 0$ and let $x, y \in X$. Since $d(x,y) = \inf_{i \in \bN} d_i(x,y)$ 
    there exists an index $m_{x,y} \in \bN$ such that for all $n \geq m_{x,y}$, $|d_n(x,y) - d(x,y)| < \e$. Now, let $N = \max \{m_{x,y} \mid x, y \in X\}$. This is well-defined because $X$ is finite. Therefore, for all $n \geq N$ and $x,y \in X$, $| d_n(x,y) - d(x,y)|< \e$ and hence $\| d_n - d\| < \e$.
\end{proof}
We can now use the above to show the desired property, by re-adapting \cite[Theorem~1]{Breugel:2012:On}. 
\begin{restatable}{lemma}{cocontinuous}\label{lem: cocontinuous} If $\mathcal{M}$ is a deterministic finite automaton, then $\Phi_{\mathcal{M}}$ is $\omega$-cocontinuous. 
\end{restatable}
\begin{proof}
 By \cref{lem:chain_conv_to_inf}, the chain $\{d_i\}_{i \in \bN}$ converges to $\inf_{i \in \bN} d_i$. Since $\Phi_{\mathcal{M}}$ is nonexpansive (\cref{lem:operator_nonexpansive}) it is also continuous (in the sense of the Banach space continuity) and therefore $\{\Phi_{\mathcal{M}}(d_i)\}_{i \in \bN}$ converges to $\Phi_{\mathcal{M}} \left(\inf_{i \in \bN} d_i\right)$. Recall that $\Phi_{\mathcal{M}}$ is monotone, which makes $\{\Phi_{\mathcal{M}}(d_i)\}_{i \in \bN}$ into a chain, which by \cref{lem:chain_pointwise_inf} and \cref{lem:chain_conv_to_inf} converges to $\inf_{i \in \bN} \{\Phi_{\mathcal{M}} (d_i)\}$. Since limit points are unique, $\inf_{i \in \bN} \{\Phi_{\mathcal{M}} (d_i)\} = \Phi_{\mathcal{M}} \left(\inf_{i \in \bN} d_i\right)$.
\end{proof}
We can combine the preceding results and provide a straightforward characterisation of the distance between languages represented by arbitrary regular expressions, denoted as $e, f \in \Exp$. Utilising a simple argument based on \Cref{prop:homomorphisms_isometries}, which asserts that automata homomorphisms are nonexpansive, one can demonstrate that the distance between $\sem{e}$ and $\sem{f}$ in the final automaton is equivalent to the distance between $[e]_{\acirel}$ and $[f]_{\acirel}$ in $\gen{[e]_{\acirel},[f]_{\acirel}}{\mathcal{Q}}$. This is the least subautomaton of $\mathcal{Q}$ that contains the derivatives (modulo $\acirel$) reachable from $[e]_{\acirel}$ and $[f]_{\acirel}$. Importantly, this automaton is finite (\Cref{lem:locally_finite}), allowing us to apply the Kleene fixpoint theorem to calculate the distance.

Let ${\Psi}_{e,f}^{(0)}$ denote the discrete metric on the set $\gen{[e]{\acirel},[f]{\acirel}}{\mathcal{Q}}$ (the top element of the lattice of pseudometrics over that set). Define ${\Psi}_{e,f}^{(n+1)} = \Phi_{\gen{[e]{\acirel},[f]{\acirel}}{\mathcal{Q}}} ({\Psi}_{e,f}^{(n)})$. Additionally, leveraging the fact that infima of decreasing chains are calculated pointwise (\Cref{lem:chain_pointwise_inf}), we can conclude with the following:
\begin{restatable}{lemma}{iterativecalc}\label{lem:iterative_calculation}
For all $e,f \in \Exp$, the underlying pseudometric of the quantitative algebra $\B$ can be given by $d^{\B}(e,f) = \inf_{i \in \bN} \left\{ {\Psi}_{e,f}^{(i)}\left([e]_{\acirel}, [f]_{\acirel}\right)\right\} $
\end{restatable}

In simpler terms, we have demonstrated that the behavioural distance between a pair of arbitrary regular expressions can be calculated as the infimum of decreasing approximations of the actual distance from above. Alternatively, one could calculate the same distance as the supremum of increasing approximations from below using the Kleene fixpoint theorem for the least fixpoint. We chose the former approach because our proof of completeness relies on the \textsf{(Cont)} rule of quantitative deduction systems. This rule essentially states that to prove two terms are at a specific distance, we should be able to prove that for all approximations of that distance from above. This allows us to replicate the fixpoint calculation through axiomatic reasoning.

\subparagraph*{Completeness result. } We start by recalling that regular expressions satisfy a certain decomposition property, stating that each expression can be reconstructed from its small-step semantics, up to $\equiv_0$. This property, often referred to as the fundamental theorem of Kleene Algebra/regular expressions (in analogy with the fundamental theorem of calculus and following the terminology of Rutten~\cite{Rutten:2000:Universal} and Silva~\cite{Silva:2010:Kleene}) is useful in further steps of the proof of completeness.
\begin{restatable}{theorem}{fundamental}{(Fundamental Theorem)}\label{thm:fundamental_theorem}
    For any $e \in \Exp$, $
    \vdash e_i \equiv_0 \sum_{a \in A} a \seq (e_i)_a + o_{\mathcal{R}}(e_i)
    $ is derivable using the axioms of \textsf{REG}.
    \end{restatable}
    The theorem above makes use of the $n$-ary generalised sum operator, which is well defined because of \textsf{(SL1-SL4)} axioms of \textsf{REG}. Let's now say that we are interested in the distance between some expressions $e,f \in \Exp$. As mentioned before, we will rely on $\gen{[e]_{\acirel},[f]_{\acirel}}{\mathcal{Q}}$, the least subautomaton of the $\acirel$ quotient of the Brzozowski automaton containing states reachable from $[e]_{\acirel}$ and $[f]_{\acirel}$. Recall that by \Cref{lem:locally_finite} its state space is finite. It turns out that the approximations from above (from \Cref{lem:iterative_calculation}) to the distance between any pair of states in that automaton can be derived through the means of axiomatic reasoning.
\begin{lemma}\label{lem:provability1}
	Let $e,f \in \Exp$ be arbitrary regular expressions and let $[g]_{\acirel}, [h]_{\acirel} \in \gen{[e]_{\acirel},[f]_{\acirel}}{\mathcal{Q}}$. 
       For all $i \in \bN$, and $\e \geq {\Psi}_{e,f}^{(i)}\left([g]_{\acirel}, [h]_{\acirel}\right)$, one can derive $\vdash g \equiv_{\e} h$ using the axioms of $\mathsf{REG}$.
\end{lemma}
\begin{proof}

We proceed by induction on $i$. For the base case, observe that ${\Psi}_{e,f}^{(0)}$ is the discrete $1$-pseudometric on the set $\gen{[e]_{\acirel},[f]_{\acirel}}{\mathcal{Q}}$ such that ${\Psi}_{e,f}^{(0)}([g]_{\acirel}, [h]_{\acirel})=0$ if and only if $g \acirel h$, or otherwise ${\Psi}_{e,f}^{(0)}([g]_{\acirel}, [h]_{\acirel})=1$. In the first case, we immediately have that $g \equiv_0 h$, because $\acirel$ is contained in distance zero axioms of \textsf{REG}. In the latter case, we can just use $(\Top)$, to show that $g \equiv_1 h $. Then, in both cases, we can apply $(\Max)$ to obtain $\vdash g \equiv_{\e} h$, since $\e \geq {\Psi}_{e,f}^{(0)}([g]_{\acirel}, [h]_{\acirel})$.
	For the induction step, let $i = j + 1$ and derive the following:
\begin{align*}
    &\e \geq {\Psi}_{e,f}^{(j + 1)}([g]_{\acirel}, [h]_{\acirel})
    \iff \e \geq \Phi_{\gen{[e]_{\acirel},[f]_{\acirel}}{\mathcal{Q}}} \left( {\Psi}_{e,f}^{(j)} \right)\left([g]_{\acirel}, [h]_{\acirel}\right) \tag{Def. of ${\Psi}_{e,f}^{j+1}$}\\
    \iff &\e \geq \max\left\{ d_{2} (o_{\mathcal{Q}}({[g]}_{\acirel}), o_{\mathcal{Q}}([h]_{\acirel})), \lambda \cdot \max_{a \in A} \left\{ {\Psi}_{e,f}^{(j)}\left({[g]_{\acirel}}_a, {[h]_{\acirel}}_{a}\right) \right\}\right\} \tag{Def. of $\Phi$}\\
    \iff &\e \geq \max\left\{ d_{2} \left(o_{\mathcal{R}}(g), o_{\mathcal{R}}(h)\right), \lambda \cdot \max_{a \in A} \left\{ \Phi_{\gen{[e]_{\acirel},[f]_{\acirel}}{\mathcal{Q}}}^{(j)}\left({[{(g)}_{a}]_{\acirel}}, {[{(h)}_{a}]_{\acirel}}\right) \right\}\right\} \tag{Def. of $\mathcal{Q}$}\\
    \iff &{\e \geq d_{2} (o_{\mathcal{R}}(g), o_{\mathcal{R}}(h))} \text{ and for all $a \in A$},~\e \cdot \lambda^{-1} \geq {\Psi}_{e,f}^{(j)}\left({[{(g)}_{a}]_{\acirel}}, {[{(h)}_{a}]_{\acirel}}\right)
\end{align*}
Firstly, since $d_2$ is the discrete $1$-pseudometric on the set $\{0,1\}$, we can use $(\Refl)$ or $(\Top)$ depending on whether $o_{\mathcal{R}}(g) = o_{\mathcal{R}}(h)$ and then apply $(\Max)$ to derive $\vdash o_{\mathcal{R}}(g) \equiv_\e o_{\mathcal{R}}(h)$. 

Let $a \in A$. We will show that $\vdash a\seq (g)_a \equiv_{\e} a \seq (h)_a $. Since $\e \cdot \lambda^{-1}$ is not guaranteed to be rational, we cannot immediately apply the induction hypothesis. Instead, we rely on \textsf{(Cont)} rule. First, pick an arbitrary rational $\e'$ strictly greater than $\e$ and fix $\{r_n\}_{n \in \bN}$ to be any decreasing sequence of rationals that converges to $\lambda^{-1}$. Let $r_N$ be an element of that sequence such that $\e' \geq \lambda \cdot \e \cdot r_N$. It is always possible to pick such element because $\{\lambda \cdot r_n\}_{n \in \bN}$ is a decreasing sequence that converges to $1$ and $\e' > \e$. Since $\epsilon \cdot r_N \geq \epsilon \cdot \lambda^{-1}$ and $\e \cdot r_{N} \in \bQ$, we can use the induction hypothesis and derive $\vdash (g)_a \equiv_{\epsilon \cdot r_{N}} (h_a)$. Then, by $(\dPref)$ axiom we have that $\vdash a \seq (g)_a \equiv_{\e'} a \seq (h)_a$. Since we have shown it for arbitrary $\e' > \e$, by \textsf{(Cont)} rule we have that $\vdash a\seq (g)_a \equiv_{\e} a \seq (h)_a $. Using \textsf{(SL5)}, we can combine all subexpressions involving the output and transition derivatives into the following:
	$$
	\vdash \sum_{a \in A} a \seq (g)_a + o_{\mathcal{R}}(g) \equiv_{\e} \sum_{a \in A} a \seq (h)_a + o_{\mathcal{R}}(h)
	$$
	Since both sides are normal forms of $g$ and $h$ existing because of \Cref{thm:fundamental_theorem}, we can apply \textsf{(Triang)} on both sides and obtain $\vdash g \equiv_{\e} h$ thus completing the proof.
\end{proof}
At this point, we have done all the hard work, and establishing completeness involves a straightforward argument that utilises the \textsf{(Cont)} rule and the lemma above.
\begin{theorem}[Completeness]
       For any $e, f \in \Exp$ and $\e \in \bQ$, if $\models_\B e \equiv_\e f$, then $\vdash e \equiv_\e f \in \mathsf{REG}$
\end{theorem}
\begin{proof}
    Assume that $ \models_\B e \equiv_\e f $, which by the definition of $\models_{\B}$ is equivalent to $d^{\B}(e,f) \leq \e$. In order to use \textsf{(Cont)} axiom to derive $\vdash e \equiv_\e f$, we need to be able to show $\vdash e \equiv_{\e'} f$ for all $\e' > \e$. Because of iterative characterisation of $d^{\B}$ from \Cref{lem:iterative_calculation}, we have that $\inf_{i \in I} \{{\Psi}_{e,f}^{(i)}([e]_{\acirel},[f]_{\acirel})\} < \e'$. Since $\e'$ is strictly above the infimum of the descending chain of approximants, there exists a point $i \in \bN$, such that $\e' > {\Psi}_{e,f}^{(i)} \left([e]_{\acirel}, [f]_{\acirel}\right)$. We can show this by contradiction. 
    
    Assume that for all $i \in \bN$, $\e' \leq {\Psi}_{e,f}^{(i)}\left([e]_{\acirel}, [f]_{\acirel}\right)$. 
     This would make $\e'$ into the lower bound of the chain $\left\{{\Psi}_{e,f}^{(i)}([e]_{\acirel}, [f]_{\acirel})\right\}_{i \in \bN}$ and in such a case $\e'$ would be less than or equal to the infimum of that chain, which by assumption is less than or equal to $\e$. By transitivity, we could obtain $\e' \leq \e$. Since $\e' > \e$, by antisymmetry we could derive that $\e'= \e$, which would lead to the contradiction. 

    Using the fact shown above, we can use \cref{lem:provability1} to obtain $\vdash e \equiv_{\e'} f \in \mathsf{REG}$ for any $\e' > \e$, which completes the proof.
\end{proof}
In simpler terms, the \textsf{(Cont)} rule enables us to demonstrate that two terms are at a specific distance by examining all strict overapproximations of that distance. Due to the iterative nature outlined in \Cref{lem:iterative_calculation}, this implies that we only need to consider finite approximants used in the Kleene fixpoint theorem. Each of those finite approximants can be axiomatically derived using \Cref{lem:provability1}.
\section{Discussion}\label{sec: conclusion}
We have presented a sound and complete axiomatisation of the shortest-distinguishing word distance between languages representing regular expressions through a quantitative analogue of equational logic~\cite{Mardare:2016:Quantitative}. Before our paper, only axiomatised behavioural distances of probabilistic/weighted transition systems existed, through (variants of) the Kantorovich/Wasserstein lifting~\cite{Larsen:2011:Metrics,DArgenio:2014:Axiomatizing,Bacci:2018:Algebraic,Bacci:2018:Bisimilarity,Bacci:2018:TV}, while we looked at a behavioural distance obtained through a more general coalgebraic framework~\cite{Baldan:2018:Coalgebraic}. Outside of the coalgebra community, the shortest-distinguishing word distance and its variants also appear in the model checking~\cite{Kwiatkowska:1990:Metric} and in the automata learning~\cite{Ferreira:2022:Tree} literature.

We have followed the strategy for proving completeness from~\cite{Bacci:2018:Bisimilarity}. The interesting insight about that strategy is that it relies on properties that are not exclusive to distances obtained through the Kantorovich/Wasserstein lifting and can be established for notions of behavioural distance for other kinds of transition systems stemming from the coalgebraic framework. In particular, one needs to show that the monotone map on the lattice of pseudometrics used in defining the distance of finite-state systems is nonexpansive with respect to the sup norm (and hence $\omega$-cocontinuous) and has a unique fixpoint, thus allowing to characterise the behavioural distance as the greatest fixpoint obtained through the Kleene fixpoint theorem. This point of view allows one to reconstruct the fixpoint calculation in terms of axiomatic manipulation involving the \textsf{(Cont)} rule, eventually leading to completeness.

We have additionally observed that in the presence of the infinitary \textsf{(Cont)} rule and the \textsf{(\dPref)} axiom, there is no need for a fixpoint rule, which is common place in all axiomatisations of regular expressions but also in other work on distances. In particular, the previous work on axiomatising a discounted probabilistic bisimilarity distance from~\cite{Bacci:2018:Bisimilarity} includes both \textsf{(\dPref)} and the fixpoint introduction rule, but its proof of completeness~\cite[Theorem~6.4]{Bacci:2018:Bisimilarity} does not involve the fixpoint introduction rule at any point. We are highly confident that in the case of that axiomatisation, the fixpoint introduction rule could be derived from other axioms in a similar fashion to the way we derived Salomaa's rule for introducing the Kleene star~\cite{Salomaa:1966:Two}. Additionally, we are interested in how much this argument relates to the recent study of fixpoints in quantitative equational theories~\cite{Mardare:2021:Fixed}. 

Moreover, the axiomatisations from~\cite{Bacci:2018:TV,Bacci:2018:Bisimilarity} rely on a slight modification of quantitative equational theories, which drop the requirement of all operations from the signature to be nonexpansive. This is dictated by the fact that the interpretation of $\mu$-recursion in Stark and Smolka's probabilistic process algebra~\cite{Stark:2000:Complete} can increase the behavioural distances in the case of unguarded recursion, while in regular expressions recursive behaviour is introduced through Kleene's star, whose interpretation is non-expansive with respect to the shortest-distinguishing-word distance. This allowed us to fit instantly into the original framework of quantitative equational theories. The earlier work~\cite{Bacci:2018:Algebraic} focusing on Markov processes~\cite{Blute:1997:Bisimulations} also relies on quantitative equational theories, but its syntax does not involve any recursive primitives. Instead, the recursive behaviour is introduced through Cauchy completion of a pseudometric induced by the axioms. The earliest works on axiomatising behavioural distances of weighted~\cite{Larsen:2011:Metrics} and probabilistic~\cite{DArgenio:2014:Axiomatizing} transition systems, studied before the introduction of quantitative equational theories, rely on ad-hoc inference systems that cannot be easily generalised. 

The pioneering works~\cite{Desharnais:2014:Metrics,Breugel:2001:Towards,Breugel:2012:On,Breugel:2006:Approximating,Breugel:2017:Probabilistic,Bacci:2019:Converging} laid foundations for behavioural (pseudo)metrics of various flavours of probabilistic transition systems. The coalgebraic point of view~\cite{Baldan:2018:Coalgebraic} allowed to generalise these ideas to a wide range of transition systems by moving from the Kantorovich/Wasserstein lifting to the abstract setting of lifting endofunctors from the category of sets to the category of pseudometric spaces. Building upon this theory, further lines of work were dedicated to asymmetric distances (called hemimetrics) through the theory of quasi-lax liftings~\cite{Wild:2022:Characteristic}, fuzzy analogues of Hennessy-Milner logic characterising behavioural distance~\cite{Goncharov:2023:Kantorovich,Beohar:2023:Hennessy}, fibrational generalisations involving quantale-enriched categories~\cite{Beohar:2023:Expressive}, up-to techniques allowing for efficient approximation of behavioural distances~\cite{Bonchi:2018:UpTo} and quantitative analogues of van Glabbek's linear-time branching-time spectrum~\cite{Forster:2023:Quantitative}.

In this paper, we have focused on the simplest and most intuitive instantiation of the coalgebraic framework in the case of deterministic automata, but the natural next step would be to generalise our results to a wider class of transition systems. A good starting point could be to consider coalgebras for \emph{polynomial} endofunctors, in the fashion of the framework of \emph{Kleene Coalgebra}~\cite{Silva:2010:Kleene}. Alternatively, it would be interesting to look at recent work on a family of process algebras parametric on an equational theory representing the branching constructs~\cite{Schmid:2022:Processes} and study its generalisations to quantitative equational theories. A related and interesting avenue for future work are equational axiomatisations of behavioural equivalence of Guarded Kleene Algebra with Tests (\textsf{GKAT})~\cite{Smolka:2020:Guarded,Schmid:2022:Processes} and its probabilistic extension (\textsf{ProbGKAT})~\cite{Rozowski:2023:Probabilistic}, whose completeness results rely on a powerful uniqueness of solutions axiom (\textsf{UA}). The soundness of \textsf{UA} in both cases is shown through an involved argument relying on equipping the transition systems giving the operational semantics with a form of behavioural distance and showing that recursive specifications describing finite-state systems correspond to certain contractive mappings. It may be more sensible, particularly for \textsf{ProbGKAT} to consider quantitative axiomatisations in the first place and give the proofs of completeness through the pattern explored in this paper.




\bibliography{lipics-v2021-sample-article}
\newpage
\appendix
\section{Preliminaries}
\adequacy*
\begin{proof}
    Let $\iota : \gen{e}{\mathcal{R}} \hookrightarrow \Exp$ be the canonical inclusion homomorphism. Composing it with $L_{\mathcal{R}}$ a unique homomorphism from $\mathcal{R}$ into the final automaton $\mathcal{L}$ yields a homomorphism $L_{\mathcal{R}} \circ \iota$ from $\gen{e}{\mathcal{R}}$ to the final automaton, which is the same as $L_{\gen{e}{\mathcal{R}}}$, since homomorphisms into the final automaton are unique. Using \cref{lem:adequacy} we can show the following: 
    \[
        \sem{e} = L_{\mathcal{R}}(e)=  L_{\mathcal{R}}(\iota(e)) = L_{\gen{e}{\mathcal{R}}}(e)
    \]
\end{proof}
\section{Behavioural distances}
\homomorphismsisometries*
\begin{proof}
    Follows from \cite[Theorem~5.23]{Baldan:2018:Coalgebraic} and \cite[Lemma~6.1]{Baldan:2018:Coalgebraic}
\end{proof}
\behaviouralpseudometric*
\begin{proof}
    \begin{enumerate}
        \item Follows from \cite[Example~5.33]{Baldan:2018:Coalgebraic}, \cite[Lemma~5.24]{Baldan:2018:Coalgebraic} and \cite[Theorem~6.10]{Baldan:2018:Coalgebraic}.
        \item Let $m, m' \in M$. Recall that by \cref{prop:homomorphisms_isometries}, $d_{\mathcal{M}}(m,m') = d_{\mathcal{L}}(L_{\mathcal{M}}(m),L_{\mathcal{M}}(m'))$. The implication $d_\mathcal{M}(m,m') = 0 \implies L_{\mathcal{M}}(m) = L_{\mathcal{M}}(m')$ follows from the fact that $d_{\mathcal{L}}$ is a metric. The converse holds because of \cite[Lemma~6.6]{Baldan:2018:Coalgebraic}.
        
        \end{enumerate}
\end{proof}
\section{Quantitative algebras}
\quantitativealgebra*
\begin{proof}
Since $d_{\mathcal{L}}$ is a pseudometric, then $d^{\B} = d_{\mathcal{L}} \circ (\llbracket - \rrbracket \times \llbracket - \rrbracket)$ is also a pseudometric. We now verify the nonexpansivity of interpretations of operations with non-zero arity. Let $e,f,g,h \in \Exp$, $d^{\B}(e,g) \leq \e$ and $d^{\B}(f,h) \leq \e$. 

\begin{enumerate}
    \item We show that $d^{\B}(e + f, g + h) \leq \e$. In the case when $\e = 0$, the proof simplifies to showing that if $\llbracket e \rrbracket = \llbracket g \rrbracket$ and $\llbracket f \rrbracket = \llbracket h \rrbracket$ then $\llbracket e + f\rrbracket = \llbracket g + h \rrbracket$, which holds immediately. 
    For the remaining case, when $\e > 0$, let $n = \lceil \log_{\lambda} \e \rceil$. Observe that in such a case, we have that $d^{\B}(e,g) \leq \lambda^n$ and $d^{\B}(f,h) \leq \lambda^n$. Using it, we can deduce that $\sem{e}$ and $\sem{g}$ (and similarly $\sem{f}$ and $\sem{h}$) agree on all words of length strictly below $n$ (because the shortest word for which they disagree is at least of length $n$). To put that more formally:
    $$\forall w \in A^{\ast} \ldotp |w| < n \implies \left(w \in \sem{e} \iff w \in \sem{g} \right) \wedge  \left(w \in \sem{f} \iff w \in \sem{h} \right) $$
    Let $w \in A^\ast$, such that $|w|<n$. We have that 
    \begin{align*}
        w \in \sem{e + f} &\iff w \in \sem{e} \cup \sem{f} \iff \left(w \in \sem {e}\right) \vee \left(w \in \sem{f}\right) \\
        &\iff \left(w \in \sem{g}\right) \vee \left(w \in \sem{h}\right) \tag{$|w| < n$} \\
        &\iff w \in \sem{g + h}
    \end{align*}
    And thus $\sem{e + f}$ and $\sem{g + h}$ agree on all words of the length below $n$ and therefore $d^{\B}(e + f, g + h) \leq \lambda^{n} \leq \e$.
    
    \item The case for $\e = 0$ holds immediately through the same line of reasoning as before, relying on well-definedness of $\diamond$ (concatenation) operation on languages. We focus on the remaining case, making the same simplification as before, that is we assume that $\sem{e}$ and $\sem{g}$ (as well as $\sem{f}$ and $\sem{h}$) agree on all word of length strictly below $n$). We show that $\sem{e \seq f}$ and $\sem{g \seq h}$ also agree on all words of the length strictly less than $n$. Let $w \in A^{\ast}$, such that $|w|<n$. We have that: 
    \begin{align*}
        w \in \sem{e \seq f} &\iff w \in \sem{e} \diamond \sem{f}\\ &\iff
        \left(\exists u,v \in A^*. w=uv \wedge w \in \sem{e} \wedge v \in \sem{f}\right)\\
        &\iff  \left(\exists u,v \in A^*. w=uv \wedge w \in \sem{g} \wedge v \in \sem{h}\right) \tag{ $|u|< n$ and $|v| < n$} \\
        &\iff w \in \sem{g} \diamond \sem{h} \iff w \in \sem {g \seq h}
    \end{align*}
    \item We use the same line of reasoning as before. Assume that $\sem{e}$ and $\sem{g}$ agree on all words of length below $n$. Let $w \in A^*$, such that $|w| < n$. We have the following:
    \begin{align*}
        w \in \sem{e^*} &\iff w \in \sem{e}^*\\ 
        &\iff w = \e \vee \left(\exists k \geq 1 \ldotp \exists u_1, \dots, u_k \in A^* \ldotp w = u_1\dots u_k \right.\\&\left. \qquad\qquad\wedge u_1 \in \sem{e} \wedge \dots \wedge u_k \in \sem{e}\right) \\
        &\iff w = \e \vee \left(\exists k \geq 1 \ldotp \exists u_1, \dots, u_k \in A^*\ldotp w = u_1\dots u_k \right.\\&\left. \qquad\qquad\wedge u_1 \in \sem{g} \wedge \dots \wedge u_k \in \sem{g}\right) \tag{$|u_1| < n, \dots, |u_k| < n$ }\\
        &\iff w \in \sem{g}^* \iff w \in \sem{g^*}
    \end{align*}
\end{enumerate}
\end{proof}
\generalisedpref*
\begin{proof}
	By induction on $e \in \Exp$. The cases when $e = 1$ and $e = (e_1)^*$ are not possible, because of the assumption that $\epsilon \notin \sem{e}$. 
	
	\fbox{$e = 0$} 
	Because of the $(\mathsf{0S})$ axiom, we can derive that $e \seq f \equiv_0 0 \equiv_0 0 \seq g \equiv_0 e \seq g$. We can show the desired conclusion, using $(\Max)$ axiom.
	
	\fbox{$e = a$} Holds immediately, because of $(\dPref)$ axiom.
	
	\fbox{$e = e_1 + e_2$} Because of the assumption, both $ \epsilon \notin \sem{e_1}$ and $ \epsilon \notin \sem{e_2}$. Using the induction hypothesis, we can derive that $\vdash e_1 \seq f \equiv_{\e'} e_1 \seq g$ and $e_2 \seq f \equiv_{\e'} e_2 \seq g$. We can apply the $(\mathsf{SL5})$ axiom to derive that $\vdash e_1 \seq f + e_2 \seq f \equiv_{\e'} e_1 \seq g + s_2 \seq g$. Finally, we can apply the $(\mathsf{D2})$ axiom to both sides through \textsf{(Triang)} and derive $\vdash (e_1 + e_2) \seq f \equiv_{\e'} (e_1 + e_2) \seq g$ as desired.
	
	\fbox{$e = e_1 \seq e_2$} Because of the assumption, $\epsilon \notin \sem{e_1}$ or $\epsilon \notin \sem{e_2}$. First, let's consider the subcase when both $\epsilon \notin \sem{e_1}$ and $\epsilon \notin \sem{e_2}$. By induction hypothesis, we have that $\vdash e_2 \seq f \equiv_{\e'} e_2 \seq g$. Since $\lambda \in ]0,1[$, we have that $\lambda \cdot \e' < \e'$. Because of that, we can apply induction hypothesis again and obtain $\vdash  e_1 \seq e_2 \seq f \equiv_{\e'} e_1 \seq e_2 \seq g$. Now, let's consider the subcase when $\epsilon \notin \sem{e_1}$, but $\epsilon \in \sem{e_2}$. Using $(\Nexp)$, we can obtain $\vdash e_2 \seq f \equiv_{\e} e_2 \seq g $. Then, since $\epsilon \notin \sem{e_1}$, we can apply the induction hypothesis and obtain $\vdash e_1 \seq e_2 \seq f \equiv_{\e'} e_1 \seq e_2 \seq g$ as desired. The remaining subcase, when  $\epsilon \notin \sem{e_2}$ but $\epsilon \in \sem{e_1}$ is symmetric and therefore omitted. 
	 

\end{proof}
\starlemma*
\begin{proof}
	By induction. If $n = 0$, then using $(\Top)$, we can immediately conclude that $\vdash g \equiv_1 e^{*} \seq f$. Since by the assumption $\e \geq \lambda^0 = 1$, we can apply \textsf{(\Max)} and obtain $\vdash g \equiv_{\e} e^* \seq f $. 
	
	For the inductive cases, we have that $\e \geq \lambda^{n+1}$ and hence $\e \cdot \lambda^{-1} \geq \lambda^n $. We cannot instantly apply the induction hypothesis, as $\e \cdot \lambda^{-1}$ is not guaranteed to be rational. Instead, we will use \textsf{(\Cont)} of quantitative deduction systems. Let $\e'$ be an arbitrary rational number strictly greater than $\e$ and let $\{r_n\}_{n \in \bN}$ be any decreasing sequence of rationals that converges to $\lambda^{-1}$. Pick element $r_N$ of that sequence that satisfies that $\e' \geq \e \cdot \lambda \cdot r_N$. We can always pick such an element, as $\{r_n\}_{n \in \bN}$ gets arbitrarily close to $\lambda^{-1}$, so $\{\lambda \cdot r_n\}_{n \in \bN}$ is a decreasing sequence that converges to $1$ and additionally we have that $\e' > \e$, so $\frac{\e'}{\e} >1$.	From the definition of the limit, we know that there exists large enough $N \in \bN$, such that $|\lambda \cdot r_{N}-1| \leq \frac{\e'}{\e} - 1$. We can simplify the above relying on the fact that $\lambda \cdot r_n \geq 1$ for all $n \in \bN$ and obtain that indeed $\e' \geq \e \cdot \lambda \cdot r_N$ as desired. 
	
		Since $\e \cdot r_N \geq \e \cdot \lambda^{-1} \geq \lambda^{n}$, we can apply induction hypothesis and obtain that $\vdash e \equiv_{\e \cdot r_N} g^* \seq f$. Since $\epsilon \notin \sem{e}$, we can now use \Cref{lem:generalised_pref} to derive that $e \seq g \equiv_{\e'} e \seq e^{*} \seq f$. Since we have shown it for arbitrary $\e' > \e$, we can use \textsf{(Cont)} rule of the quantitative deduction systems and conclude that $\vdash e \seq g \equiv_{\e} e \seq e^* \seq f$, as desired.

		Then, because of $(\Refl)$, we have that $\vdash f \equiv_0 f$. We can combine those two quantitative inferences using $(\mathsf{SL5})$ axiom in order to get $\vdash e \seq g + f \equiv_{\e} e \seq e^* \seq f  + f $. By assumption, the left hand side satisfies that $\vdash g \equiv_0 e \seq g + f$. Now, consider the right hand side of that quantitative inference:
	\begin{align*}
		\vdash e \seq e^{*} \seq f + f &\equiv_0 e \seq e^* \seq f + 1 \seq f \tag{$\mathsf{1S}$} \\
		&\equiv_0 (e \seq e^* + 1) \seq f \tag{$\mathsf{D2}$} \\
		&\equiv_0 e^{*} \seq f \tag{$\mathsf{Unroll}$}
	\end{align*}
	We can combine the reasoning above and conclude (using $\Triang$) that $\vdash g \equiv_{\e} e^* \seq f$.
\end{proof}
\section{Completeness}
\iterativecalc*
\begin{proof}
    Recall that $d^{\B} = d_{\mathcal{L}} \circ (\sem{-} \times \sem{-})$. Moreover, the canonical quotient map $[-]_{\acirel} : \Exp \to \aciq$ is an automaton homomorphism from $\mathcal{R}$ to $\mathcal{Q}$. Composing it with a language assigning homomorphism $L_{\mathcal{Q}} : \aciq \to \pset(A^\ast)$ yields an automaton homomorphism $L_{\mathcal{Q}} \circ [-]_{\acirel} : \Exp \to \pset(A^\ast)$, which by finality must be the same as $L_{\mathcal{R}} : \Exp \to \pset(A^\ast)$, and thus (by \cref{lem:adequacy}) the same as $\sem{-}$. Using the fact that automata homomorphisms are isometric mappings between pseudometric spaces obtained through constructing behavioural pseudometrics (\cref{prop:homomorphisms_isometries}), we can derive the following:
    \begin{align*}
    d^{\B} &= d_{\mathcal{L}} \circ (\sem{-} \times \sem{-})\\
    &= d_{\mathcal{L}} \circ ((L_{\mathcal{Q}} \circ ([-]_{\acirel}) \times (L_{\mathcal{Q}} \circ ([-]_{\acirel}))\\
    &= d_{\mathcal{L}} \circ (L_{\mathcal{Q}} \times L_{\mathcal{Q}}) \circ ([-]_{\acirel} \times [-]_{\acirel})\\
    &= d_{\mathcal{Q}} \circ ([-]_{\acirel} \times [-]_{\acirel}) \tag{\cref{prop:homomorphisms_isometries}}
    \end{align*}
    Additionally, since $\gen{[e]_{\acirel},[f]_{\acirel}}{\mathcal{Q}}$ is the subautomaton of $\mathcal{Q}$ containing all the derivatives (modulo $\acirel$) of $e$ and $f$, the canonical inclusion map $\iota : \gen{[e]_{\acirel},[f]_{\acirel}}{\mathcal{Q}} \hookrightarrow \mathcal{Q}$ is an automaton homomorphism. Because $\iota([e]_{\acirel})=[e]_{\acirel}$ and $\iota([f]_{\acirel})=[f]_{\acirel}$, we can again use \cref{prop:homomorphisms_isometries} to show that
    $$
    d^{\B}(e,f) = d_{\mathcal{Q}}([e]_{\acirel}, [f]_{\acirel}) = d_{\gen{[e]_{\acirel},[f]_{\acirel}}{\mathcal{Q}}}([e]_{\acirel}, [f]_{\acirel})
    $$
    Because of the fact that $([e]_{\acirel}, [f]_{\acirel})$ has finitely many states (\cref{lem:locally_finite}) then by \Cref{lem: uniquefp}, \Cref{lem: cocontinuous} and \Cref{thm: kleene} one can use the simplified iterative formula to calculate the behavioural pseudometric on $\gen{[e]_{\acirel},[f]_{\acirel}}{\mathcal{Q}}$.
\end{proof}

\fundamental*
\begin{proof}
We proceed by induction on $e \in \Exp$. The base cases are trivial, so we just demonstrate the case when $e = a$ for $a \in A$. 

    \fbox{$e = a$} 
    \begin{align*}
       \vdash a &\equiv_{0} a \seq 1 \tag{$\mathsf{S1}$}\\
        &\equiv_0 a \seq 1 + 0 \tag{$\mathsf{SL4}$}\\
        &\equiv_0 a \seq (a)_a + o_{\mathcal{R}}(a) \tag{Def. of derivatives}
    \end{align*}
    Now, observe that for all $a' \in A \setminus \{a\}$, we have that $(a)_a' = 0$. Using axiom ($\mathsf{S0}$), $a' \seq (a)_{a'} \equiv_0 0$. Through induction on the size of $A \setminus \{a\}$, using axioms ($\mathsf{SL1}$) and ($\mathsf{SL4}$), one can show that $\sum_{a' \in A \setminus \{a \}} a' \seq (a)_{a'} \equiv_0 0$. We can now combine the intermediate results into the following:
    \begin{align*}
       \vdash a &\equiv_0 a \seq (a)_a + o_{\mathcal{R}}(a) \tag{Previous derivations} \\
        &\equiv_0  a \seq (a)_a + o_{\mathcal{R}}(a) + 0 \tag{$\mathsf{SL1}$} \\
        &\equiv_0 a \seq (a)_a + 0 + o_{\mathcal{R}}(a) \tag{$\mathsf{SL2}$} \\
        &\equiv_0 a \seq (a)_a + \sum_{a' \in A \setminus \{a\}} a' \seq (a)_{a'} + o_{\mathcal{R}}(a) \tag{Previous inductive argument} \\
        &\equiv_0 \sum_{a' \in A} a' \seq (a)_{a'} + o_{\mathcal{R}}(a) \tag{Def. of $n$-ary sum}
    \end{align*}

    \fbox{$e = f + g$}
    \begin{align*}
        \vdash f + g &\equiv_0 \left(\sum_{a \in A} a \seq (f)_a + o_{\mathcal{R}}(f)\right) + \left(\sum_{a \in A} a \seq (g)_a + o_{\mathcal{R}}(g)\right) \tag{Induction hypothesis} \\
        &\equiv_0 \sum_{a \in A}\left( a \seq (f)_a + a \seq (g)_a\right) + \left( o_{\mathcal{R}}(f) + o_{\mathcal{R}}(g)\right) \tag{$\mathsf{S3}$}\\
        &\equiv_0 \sum_{a \in A}\left( a \seq \left((f)_a + (g)_a\right)\right) + \left( o_{\mathcal{R}}(f) + o_{\mathcal{R}}(g)\right) \tag{$\mathsf{D1}$}\\
        &\equiv_0 \sum_{a \in A}\left( a \seq \left((f + g)_a\right)\right) + o_{\mathcal{R}}(f + g)  \tag{Def. of derivatives}\\
    \end{align*}

    \fbox{$e = f \seq g $}
    \begin{align*}
        \vdash f \seq g &\equiv_0 \left(\sum_{a \in A} a \seq (f)_a + o_{\mathcal{R}}(f)\right) \seq g \tag{Induction hypothesis} \\
        &\equiv_0 \sum_{a \in A} a \seq (f)_a \seq g  + o_{\mathcal{R}}(f) \seq g \tag{$\mathsf{D2}$} \\
        &\equiv_0 \sum_{a \in A} a \seq (f)_a \seq g  + o_{\mathcal{R}}(f) \seq \left(\sum_{a \in A} a \seq (g)_a + o_{\mathcal{R}}(g) \right) \tag{Induction hypothesis} \\
        &\equiv_0 \sum_{a \in A} a \seq (f)_a \seq g  +  \left(\sum_{a \in A}  o_{\mathcal{R}}(f) \seq a \seq (g)_a + o_{\mathcal{R}}(f) \seq o_{\mathcal{R}}(g) \right) \tag{$\mathsf{D1}$} \\
        &\equiv_0 \sum_{a \in A} a \seq (f)_a \seq g  +  \left(\sum_{a \in A} a \seq  o_{\mathcal{R}}(f) \seq (g)_a + o_{\mathcal{R}}(f) \seq o_{\mathcal{R}}(g) \right) \tag{($\mathsf{1S}$) and ($\mathsf{S1}$) if $o_{\mathcal{R}}(f)=1$ or ($\mathsf{0S}$) and ($\mathsf{S0}$) if $o_{\mathcal{R}}(f)=0$} \\
        &\equiv_0 \sum_{a \in A} \left(a \seq (f)_a \seq g  +  a \seq  o_{\mathcal{R}}(f) \seq (g)_a \right) + o_{\mathcal{R}}(f) \seq o_{\mathcal{R}}(g) \tag{$\mathsf{SL3}$}\\
        &\equiv_0 \sum_{a \in A} a \seq \left((f)_a \seq g  +  o_{\mathcal{R}}(f) \seq (g)_a \right) + o_{\mathcal{R}}(f) \seq o_{\mathcal{R}}(g) \tag{$\mathsf{D1}$}\\
        &\equiv_0 \sum_{a \in A} a \seq \left(f \seq g\right)_a + o_{\mathcal{R}}(f \seq g) \tag{Def. of derivatives}\\
    \end{align*}

    \fbox{$e = f^\ast$}
    \begin{align*}
        \vdash f^{\ast} &\equiv_0 \left( \sum_{a \in A} a \seq (f)_a + o_{\mathcal{R}}(f)\right)^{\ast} \tag{Induction hypothesis}\\
        &\equiv_0 \left(\sum_{a \in A} a \seq (f)_a\right)^{\ast} \tag{($\mathsf{Tight}$) if $o_{\mathcal{R}}=1$ or ($\mathsf{SL4}$) if $o_{\mathcal{R}}=0$} \\
        &\equiv_0  \left(\sum_{a \in A} a \seq (f)_a\right) \seq  \left(\sum_{a \in A} a \seq (f)_a\right)^{\ast} + 1 \tag{$\mathsf{Unroll}$}\\
        &\equiv_0  \left(\sum_{a \in A} a \seq (f)_a\right) \seq  f^{\ast} + 1 \tag{Steps 1-2 and $(\Nexp)$} \\
        &\equiv_0 \left(\sum_{a \in A} a \seq (f)_a \seq  f^{\ast} \right)  + 1 \tag{$\mathsf{D2}$}\\
        &\equiv_0 \left(\sum_{a \in A} a \seq \left(f^{\ast}\right)_a \right)  + o_{\mathcal{R}}\left(f^{\ast}\right) \tag{Def. of derivatives}\\
    \end{align*}
    
\end{proof}


\end{document}